\documentclass[a4paper,11pt]{article}
\usepackage{jheppub} 
\usepackage{lineno}

\usepackage{graphicx}
\usepackage{bm}
\usepackage{bbold}
\usepackage{physics}
\usepackage{xcolor}
\usepackage{enumitem}
\usepackage{amsmath, amssymb, amsthm}
\usepackage[normalem]{ulem}
\usepackage{physics}
\usepackage{comment}
\usepackage{mathtools}

\newtheorem{conjecture}{Conjecture}
\newtheorem{proposition}{Proposition}
\newtheorem{fact}{Fact}[section]
\newtheorem{lemma}{Lemma}[section]

\newcommand{\nocontentsline}[3]{}
\newcommand{\tocless}[2]{\bgroup\let\addcontentsline=\nocontentsline#1{#2}\egroup}

\newcommand{\CX}{\textsf{C}X}
\newcommand{\CCX}{\textsf{CC}X}
\newcommand{\CNX}[2]{\textsf{C}^{(N-1)}_{#1} {X}_{#2}}
\newcommand{\perm}[1]{\mathcal{S}_{#1}}
\newcommand{\IP}{\ensuremath{\textsf{IP}}}
\newcommand{\RW}{\ensuremath{\textsf{RW}}}
\newcommand{\EX}[1]{\ensuremath{\textsf{EX}(#1)}}
\newcommand{\poly}[1]{\ensuremath{{\rm poly}(#1)}}
\newcommand{\pseudothermtime}[1]{\ensuremath{t_{\rm indist}^{(#1)}}}

\arxivnumber{2403.09619} 

\usepackage[colorlinks=true]{hyperref}
\hypersetup{citecolor = blue, linkcolor=blue, urlcolor=blue}

\title{Dynamics of Pseudoentanglement}

\author{Xiaozhou Feng}
\author{Matteo Ippoliti}
\affiliation{Department of Physics, The University of Texas at Austin, Austin, TX 78712, USA}

\abstract{
The dynamics of quantum entanglement plays a central role in explaining the emergence of thermal equilibrium in isolated many-body systems. However, entanglement is notoriously hard to measure, and can in fact be ``forged'': recent works have introduced a notion of {\it pseudoentanglement} describing ensembles of many-body states that, while only weakly entangled, cannot be efficiently distinguished from states with much higher entanglement, such as random states in the Hilbert space.
This prompts the question: how much entanglement is truly necessary to achieve thermal equilibrium in quantum systems? 
In this work we address this question by introducing random circuit models of quantum dynamics that, at late times, equilibrate to pseudoentangled ensembles---a phenomenon we name {\it ensemble pseudothermalization}. These models replicate all the {efficiently observable} predictions of thermal equilibrium, while generating only a small (and tunable) amount of entanglement, thus deviating from the ``maximum-entropy principle'' that underpins thermodynamics. 
We examine 
(i) how a pseudoentangled ensemble on a small subsystem spreads to the whole system as a function of time, and 
(ii) how a pseudoentangled ensemble can be generated from an initial product state. 
We map the above problems onto a family of classical Markov chains on subsets of the computational basis. The mixing times of such Markov chains are related to the time scales at which the states produced from the dynamics become indistinguishable from Haar-random states at the level of each statistical moment, or number of copies. 
Based on a combination of rigorous bounds and conjectures supported by numerics, we argue that each Markov chain's relaxation time and mixing time have different asymptotic behavior in the limit of large system size. This is a necessary condition for a {\it cutoff phenomenon}: an abrupt dynamical transition to equilibrium. 
We thus conjecture that our random circuits give rise to asymptotically sharp distinguishability transitions.
}

\begin{document}

\maketitle



\section{Introduction.}\label{sec:intro}

\begin{figure*}
    \centering
    \includegraphics[width=0.99\textwidth]{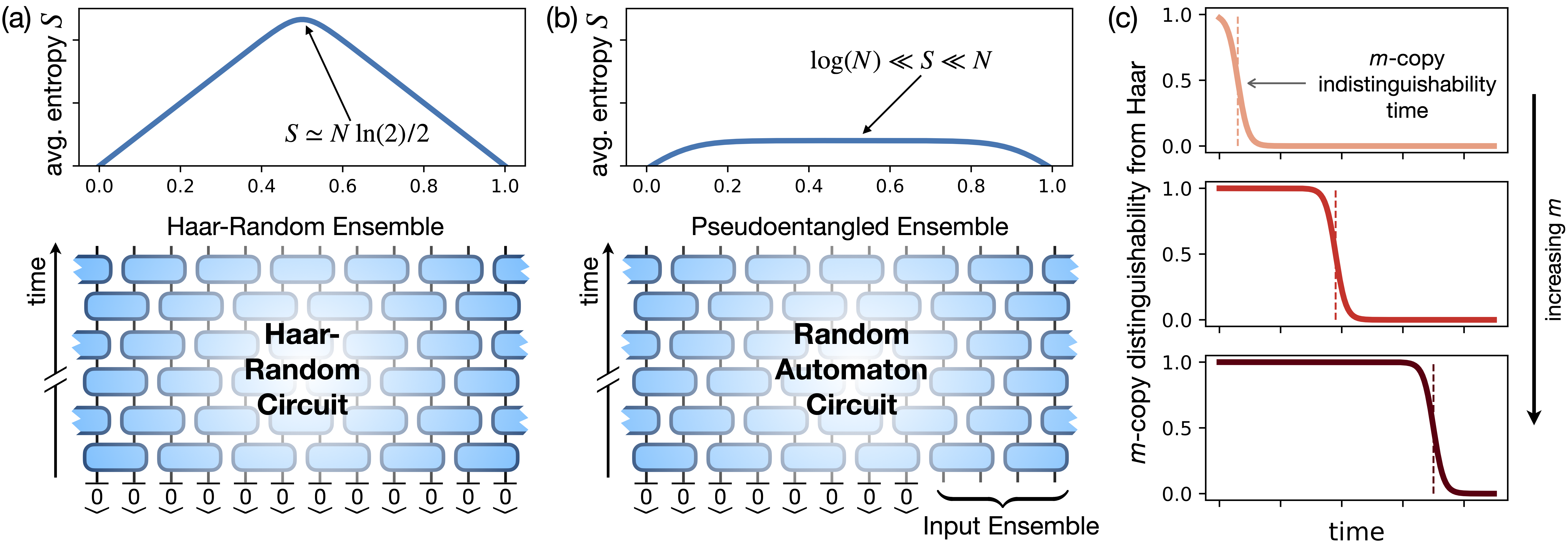}
    \caption{Main ideas of this work. 
    (a) Quantum thermalization: chaotic quantum dynamics on an isolated many-body systems (modeled by a Haar-random unitary circuit, bottom) produces random states at late times, with near-maximal entanglement entropy described by the Page curve (top).
    (b) In analogy with thermalization, we consider families of random `automaton' circuits (bottom, see text) that, acting on suitable ensembles of initial states, produce {\it pseudoentangled} state ensembles at late times. These have relatively low entanglement entropy (top), but are computationally indistinguishable from Page states. 
    (c) We study the dynamical process by which pseudoentanglement is generated and/or propagated over time, quantified by the distinguishability of the output state ensembles from the Haar ensemble: can an observer, by making joint measurements on $m$ copies of an output state, {\it efficiently} tell it apart from a random state? For each value of $m$, this defines a time scale dubbed {\it $m$-copy indistinguishability time} (vertical dashed lines). This time scale may become sharp in the thermodynamic limit if the problem features a cutoff phenomenon. }
    \label{fig:main_ideas}
\end{figure*}

Entanglement is a fundamental concept in quantum information science and in quantum many-body physics: it is at once a resource for quantum advantage in information processing~\cite{horodecki_quantum_2009,chitambar_quantum_2019} and a framework to characterize different phases of matter beyond conventional order parameters~\cite{kitaev_topological_2006,li_entanglement_2008,jiang_identifying_2012}.
Furthermore, the generation and propagation of entanglement in many-body systems is a key ingredient in our understanding of how thermal equilibrium emerges from the otherwise reversible unitary dynamics of isolated systems~\cite{rigol_thermalization_2008,nandkishore_many-body_2015,kaufman_quantum_2016,mezei_black_2020}. 
Due to entanglement, local subsystems may be described by mixed states even when the global state is pure; these local mixed states can equilibrate to universal forms~\cite{alhambra_quantum_2022,vidmar_generalized_2016}, thus recovering the predictions of thermodynamics.  
In the absence of any conservation laws (for example in driven systems without symmetries), the predicted equilibrium is simply the maximally-mixed state $\rho_{A,{\rm th}} \propto \mathbb{I}_A$ in each local subsystem $A$, indicating saturation of entanglement to its maximal value. 
Informally, at late times the state of the system is expected to resemble a typical Haar-random state in the Hilbert space\footnote{Convergence to Haar-random states can be shown rigorously in systems modeled by random unitary circuits~\cite{nahum_quantum_2017,bertini_entanglement_2019,zhou_entanglement_2020,brandao_models_2021,fisher_random_2023,potter_entanglement_2022,mittal_local_2023}, and even in some quasi-periodically driven systems~\cite{pilatowsky-cameo_complete_2023,pilatowsky-cameo_hilbert-space_2024}.}, with a characteristic extensive scaling of entanglement known as the Page curve~\cite{Page_average_1993}, sketched in Fig.~\ref{fig:main_ideas}(a).

Despite its ubiquity and far-reaching importance, entanglement is generally difficult to control, quantify, and measure~\cite{horodecki_quantum_2009}, especially beyond few-body settings~\cite{amico_entanglement_2008}. 
The last decade has seen substantial progress in these directions, particularly thanks to experimental advances in quantum simulators and theoretical developments in randomized measurements~\cite{islam_measuring_2015,kaufman_quantum_2016,elben_renyi_2018,brydges_probing_2019,huang_predicting_2020,elben_randomized_2023}; nonetheless, the cost of measuring the entanglement entropy of an $N$-qubit system remains exponential in $N$ in the worst case, ultimately limiting the reach of all these approaches. 

Recent works~\cite{aaronson_quantum_2023,giurgica-tiron_pseudorandomness_2023,jeronimo_pseudorandom_2024} highlight the intrinsic subtlety of measuring many-body entanglement in a dramatic way, by introducing a notion of {\it pseudoentanglement}---a forgery of entanglement that cannot be efficiently distinguished from the real thing.
More specifically, these works construct ensembles of quantum states that, while only weakly entangled, are indistinguishable from random states (whose entropy tracks the Page curve and is thus near-maximal) to any observer with a polynomial amount of resources (time, memory, etc). 
These ensembles are interesting from the point of view of computation and cryptography~\cite{ji_pseudorandom_2018,lu_quantum_2023}, and can be efficiently constructed through the use of quantum-secure pseudorandom permutations~\cite{aaronson_quantum_2023,zhandry_note_2016}. 

Motivated by these new developments, in this work we reexamine the role of entanglement in quantum thermalization. 

\subsection{Summary of results}

We introduce constrained models of time evolution whose late-time equilibrium is not described by the highly-entangled Page states usually associated to high-temperature thermalization, but rather by pseudoentangled state ensembles. These ensembles are weakly entangled, yet they reproduce all the {\it efficiently verifiable} predictions of thermal equilibrium. We term this condition (and the dynamical processes that generate it over time) {\it ``ensemble pseudothermalization''}. 
Depending on one's point of view, this may be viewed as a well-disguised imitation of thermal equilibrium, or as a perfectly valid instance of it. 
The former perspective can be argued on first principles: pseudoentangled state ensembles by definition do not obey the ``maximum entropy principle'' that underlies thermodynamics\footnote{More precisely, they do obey a maximum entropy principle, but only within a restricted space characterized by low entanglement.} and are not stable against generic perturbations.
The latter perspective can be argued on operational grounds: if a deviation from thermal equilibrium cannot be detected {\it efficiently}, it should not count. 
It is also worth noting that this type of equilibration happens at the level of an ensemble, for example over realizations of a random circuit, rather than at the level of individual quantum states as in conventional quantum thermalization~\cite{srednicki_chaos_1994,rigol_thermalization_2008,nandkishore_many-body_2015}. Pseudoentanglement is an ensemble property; there is no sense in which a single quantum state can be pseudoentangled. 

Understanding the process of ensemble pseudothermalization requires studying the dynamical generation and propagation of pseudoentanglement in many-body systems. This is analogous to the role of entanglement dynamics in thermalization (with the above caveat about the necessity of an ensemble). Like in that case, {\it random circuit} models of time evolution prove very useful~\cite{nahum_quantum_2017,fisher_random_2023,potter_entanglement_2022}.
We introduce random circuit models that prepare pseudoentangled state ensembles as their equilibrium steady states, much like the Page states for generic quantum dynamics. This is sketched in Fig.~\ref{fig:main_ideas}(b).
Specifically, we consider the random {\it subset-phase state} ensemble~\cite{aaronson_quantum_2023} and the random {\it subset state} ensemble~\cite{giurgica-tiron_pseudorandomness_2023,jeronimo_pseudorandom_2024}. 
We introduce families of random reversible classical circuits (also known as `automaton' quantum circuits\footnote{Automaton quantum circuits should not be confused with `quantum cellular automata', which are general locality-preserving unitaries~\cite{farrelly_review_2020}.} in the condensed matter literature~\cite{gopalakrishnan_operator_2018,iaconis_anomalous_2019,iaconis_quantum_2021,klobas_exact_2021}) that have pseudoentangled state ensembles as their late-time fixed points, and use them to quantitatively study the process of ensemble pseudothermalization.

We focus on two scenarios modeled after paradigmatic problems in `true' quantum thermalization: 
(i) the spreading of pseudoentanglement from a small subsystem placed in contact with a larger, disentangled system, and 
(ii) the generation of pseudoentanglement from an initial product state. 
The key technical step in both cases is a mapping of pseudoentanglement generation onto the equilibration of certain classical Markov chains over subsets of the computational basis---with the size of subsets related to the number of state copies $m$ available to the observer. This mapping allows us to make contact with a rich body of mathematical work on the equilibration of Markov chains on finite spaces, including the so-called {\it cutoff phenomenon}~\cite{diaconis_generating_1981,diaconis_cutoff_1996,chen_cutoff_2008,lacoin_cutoff_2011,chen_cutoff_2021}. This phenomenon refers to the sudden equilibration that can be seen in certain random walks, often over highly-connected graphs (for example permutations of a deck of cards induced by a riffle shuffle~\cite{aldous_shuffling_1986}), in contrast with e.g. diffusion on a Euclidean lattice where equilibration is smooth and characterized by a dynamical scaling collapse $\sim t/L^z$, $z = 2$. 

Based on this connection, we argue that ensemble pseudothermalization in our random circuit models happens in an asymptotically sharp fashion: our output states are easily distinguishable from random states for a long time, until they abruptly (i.e., over a {parametrically shorter} window of time) become indistinguishable at the level of $m$ copies, Fig.~\ref{fig:main_ideas}(c).
This dynamical transition in the distinguishability of low-entanglement states from random states is a new fundamental phenomenon in the dynamics of isolated quantum many-body systems.
It enriches and complements our growing understanding of thermal equilibrium at the level of higher statistical moments~\cite{roberts_chaos_2017,hunter-jones_unitary_2019,cotler_emergent_2023,ippoliti_solvable_2022,pilatowsky-cameo_complete_2023,fava_designs_2023} and opens a number of exciting directions for future work.

Our results imply several interesting consequences for the hardness of learning complex features of quantum dynamics, such as entanglement growth and measurement-induced phenomena, and pose limits on the existence of general, efficient protocols to address these problems.

\subsection{Structure of the paper}

The rest of the paper is organized as follows. 
In Sec.~\ref{sec:review} we provide a brief, self-contained review of some useful background: pseudoentanglement (Sec.~\ref{sec:review_pe}), Markov chains on finite sets and the cutoff phenomenon (Sec.~\ref{sec:review_cutoff}), and certain random walks over permutation groups (Sec.~\ref{sec:review_ip_rw}). 
We then introduce the general framework of ensemble pseudothermalization via automaton circuit dynamics in Sec.~\ref{sec:automaton}, and apply it to the problem of spreading and generation of pseudoentanglement in Sec.~\ref{sec:spreading_pe} and \ref{sec:generation} respectively. 
Finally we summarize our results and discuss their practical and fundamental implications, as well as outstanding questions and future directions, in Sec.~\ref{sec:discussion}. 


\section{Review of key concepts}\label{sec:review}

\subsection*{Asymptotic notation}

In this work we make extensive use of the following asymptotic notation. Given two functions $f,g:\mathbb{N} \mapsto \mathbb{R}^+$, we write:
\begin{itemize}
    \item $f = o(g)$ if $\lim_{N \to \infty} f(N) / g(N) = 0$;
    \item $f = O(g)$ if $f(N) < C g(N)$ for all $N$ for some constant $C>0$;
    \item $f = \Theta(g)$ if $f = O(g)$ and $g = O(f)$;
    \item $f = \Omega(g)$ if $g = O(f)$;
    \item $f = \omega(g)$ if $\lim_{N \to \infty} f(N) / g(N) = \infty $.
\end{itemize}
With a slight abuse of notation, we also write asymptotic upper and lower bounds as e.g. $\Omega(f) \leq g \leq O(h)$ to denote that $g = O(h)$ and $g = \Omega(f)$. We also use a tilde to denote when one of these conditions applies up to a ${\rm polylog}(N)$ factor: e.g. we may write $f \leq O(g(N)/\log(N))$ as $f \leq \tilde{O}(g)$. 

\subsection{Pseudoentanglement} \label{sec:review_pe}
An ensemble $\mathcal{E}$ of quantum states on $N$ qubits is said to be {\it pseudoentangled}~\cite{aaronson_quantum_2023} 
if the states $\ket{\psi} \in \mathcal{E}$ meet the following criteria\footnote{Ref.~\cite{aaronson_quantum_2023} gives a more general definition based on two ensembles with different entanglement that are computationally indistinguishable from each other. Here we always take the second ensemble to be the Haar-random ensemble.} (to be specified precisely below):
\begin{enumerate}[label=(\roman*)]
\item they are efficiently preparable;
\item they have low entanglement;
\item they are {\it computationally pseudorandom}.
\end{enumerate}
Point (i)  means there should be a $\poly{N}$-depth circuit that prepares each state in the ensemble. 
Point (ii) for our purposes will mean that there exists a bipartition of the system into extensive parts $A$, $\bar{A}$ such that $S(\Tr_{\bar A}\ketbra{\psi}) = o(N)$ with high probability over $\psi\sim\mathcal{E}$---i.e., the states have {\it sub-volume-law} entanglement\footnote{The definition in Ref.~\cite{aaronson_quantum_2023} demands an `entanglement gap' for every cut, but from a physics perspective, it is reasonable to require an entanglement gap only for spatially local bipartitions of the system (e.g., a `pseudo-area-law'). This is not important in the present work since we use random subset-(phase-)states which have an entanglement gap on all cuts.}.
The crux of the definition is point (iii), computational pseudorandomness, which can be stated as follows.
Defining the $m$-th moment operator, or average $m$-copy state,
\begin{equation}
\rho^{(m)}_{\mathcal E} = \mathbb{E}_{\psi \sim \mathcal E}[\ketbra{\psi}^{\otimes m}],
\label{eq:moment_op}
\end{equation}
the ensemble $\mathcal{E}$ is computationally pseudorandom if, for any efficient quantum algorithm $\mathcal{A}$, we have 
\begin{equation} 
\left |\mathcal{A}\left(\rho_{\mathcal E}^{(m)} \right) - \mathcal{A} \left(\rho_{{\rm Haar}}^{(m)} \right) \right| = o \left(1/\poly{N} \right)
\label{eq:pe_criterion_algo}
\end{equation} 
whenever $m = O(\poly{N})$. In words, this means that any polynomial-time observation on a polynomial number of state copies cannot distinguish $\mathcal{E}$ from the Haar ensemble by more than a superpolynomially small amount, $\epsilon = o(1/\poly{N})$. Statistically resolving this discrepancy between the two ensembles would require at least $\Omega(1/\epsilon^2) = \omega(\poly{N})$ repetitions of the experiment. Thus, with a polynomially bounded amount of time and memory, nothing the observer can do would successfully distinguish $\mathcal{E}$ from the Haar distribution.

In this work we will drop the restriction on the efficiency of the algorithm $\mathcal{A}$, allowing arbitrary measurements on the $m$ copies. Then Eq.~\eqref{eq:pe_criterion_algo} reduces to
\begin{equation} 
D_{\rm tr} \left( \rho_{\mathcal E}^{(m)}, \rho_{{\rm Haar}}^{(m)} \right)  = o \left(1/\poly{N} \right),
\label{eq:pe_criterion_tracedistance}
\end{equation} 
where $D_{\rm tr}(\rho,\sigma) = \frac{1}{2} \| \rho-\sigma\|_{\rm tr}$ is the trace distance and $\| A\|_{\rm tr} = {\rm Tr}(\sqrt{A^\dagger A})$ is the trace norm. 
The trace distance  characterizes the distinguishability between two quantum states: if $D_{\rm tr}(\rho,\sigma) = \epsilon$, then there is an observable\footnote{
One can write the observable as $\mathcal{O} = \Pi_+ - \Pi_-$, where $\Pi_{\pm}$ are the projectors on positive (negative) eigenvalue subspaces of $\rho-\sigma$.
} 
$\mathcal O$, with $-I \leq \mathcal O\leq I$, that gives ${\rm Tr}(\mathcal O \rho) - {\rm Tr}(\mathcal O \sigma) = 2\epsilon$. An observer can tell apart $\rho$ from $\sigma$ by measuring $\mathcal{O}$ a number of times $\sim 1/\epsilon^2$. This is the optimal way to distinguish two states~\cite{nielsen_quantum_2020}.
Eq.~\eqref{eq:pe_criterion_tracedistance} describes ``information-theoretic'' pseudorandomness, a condition stronger than computational pseudorandomness. It describes an ensemble $\mathcal{E}$ that cannot be distinguished from the Haar ensemble by making a polynomial number of experiments on polynomially many copies, even if the experiments are allowed arbitrary complexity (e.g. projective measurements in a basis prepared by an exponentially deep circuit).

A simple pseudoentangled ensemble is given by random\footnote{
If $f$ and $S$ are taken to be genuinely random, then the ensemble satisfies Eq.~\eqref{eq:pe_criterion_tracedistance} (information-theoretic pseudorandomness), but is not efficiently preparable. However one can replace them by suitable pseudorandom functions and permutations, respectively, in such a way that the ensemble is efficiently preparable and satisfies Eq.~\eqref{eq:pe_criterion_algo} (computational pseudorandomness).
} 
{\it subset phase states}~\cite{aaronson_quantum_2023},
\begin{equation}
    \ket{\psi_{S,f}} = \frac{1}{|S|^{1/2}} \sum_{\mathbf z\in S} (-1)^{f(\mathbf{z})}\ket{\mathbf z} ,
    \label{eq:subsetphase_state}
\end{equation}
where $f$ is a random Boolean function on $\{0,1\}^N$, and $S$ is a subset of bitstrings, $S\subset\{0,1\}^N$, chosen uniformly at random out of all subsets of some fixed cardinality $|S| = K$. We call this space of subsets $\Sigma_K$. To ensure pseudoentanglement, $K$ must scale suitably with the system size $N$: $\omega(\poly{N}) \leq K \leq o(2^{N/2})$. The role of $K$ is to determine the actual entanglement of the states in the ensemble: we have the average von Neumann entropy
\begin{equation} 
S_{\rm vN} (A) \sim \min \{|A|\log 2,\log K\}
\end{equation}
for all subsystems $A$ (regardless of locality), due to the fact that $K$ upper bounds the Schmidt rank of $\ket{\psi_{S,f}}$ on any bipartition of the $N$ qubits. 
Thus the restriction on $K\geq \omega(\poly{N})$ ensures that, e.g., the half-cut entanglement entropy can scale faster than $\log N$ (necessary so that the states cannot be efficiently distinguished from Haar-random ones by measuring their purity), and the restriction $K \leq 2^{o(N)}$ ensures sub-volume-law scaling [see Fig.~\ref{fig:main_ideas}(a,b)]. 

Very recently, Refs.~\cite{giurgica-tiron_pseudorandomness_2023,jeronimo_pseudorandom_2024} showed that the random phase factors are not necessary to have a pseudoentangled ensemble. Random \textit{subset states} suffice:
\begin{equation}
      \ket{\psi_S} = \frac{1}{|S|^{1/2}} \sum_{\mathbf z\in S} \ket{\mathbf z},
    \label{eq:subset_state}
\end{equation}
with $S$ drawn uniformly at random from $\Sigma_K$ (subsets of the computational basis of cardinality $K$).
This ensemble obeys
\begin{equation}
    D_{\rm tr} \left( \rho_{\mathcal{E}}^{(m)}, \rho_{{\rm Haar}}^{(m)} \right) \le O \left(\frac{m^2}{K} \right) + O \left( \frac{mK}{D} \right),
\end{equation}
with $D = 2^N$ the Hilbert space dimension.
This is superpolynomailly small if $\omega(\poly{N}) \leq K \leq o(2^N / \poly{N})$ and $m = O(\poly{N})$. 

\subsection{Equilibration of Markov chains and cut-off phenomenon}\label{sec:review_cutoff}

\begin{figure}
    \centering
    \includegraphics[width=0.5\columnwidth]{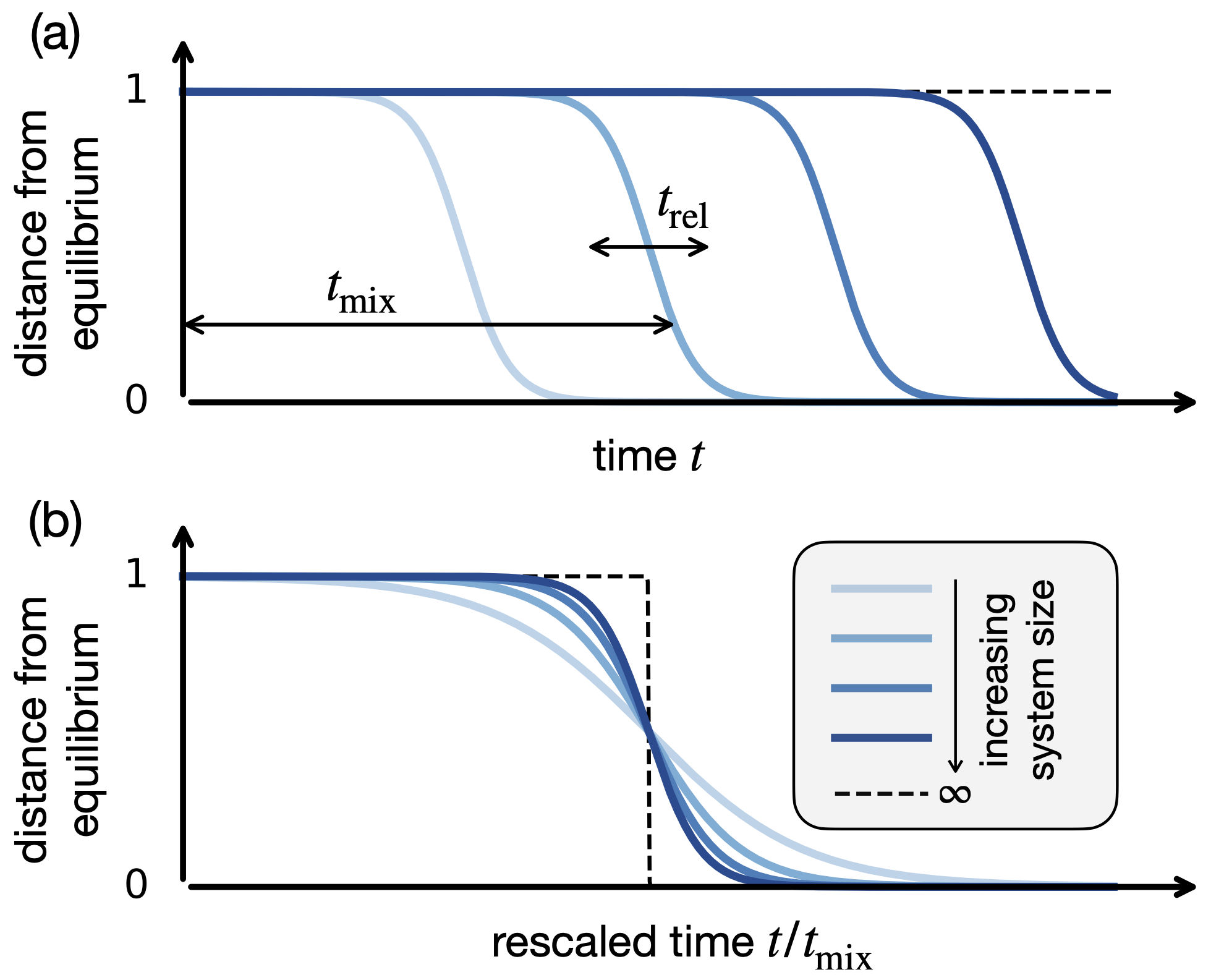}
    \caption{Schematic of the cutoff phenomenon for Markov chains on finite sets. 
    (a) A family of Markov chains on systems of size $N$ has a cutoff if it exhibits ``sudden'' equilibration as $N\to\infty$, in the following sense: the total variation distance from equilibrium remains close to maximal for a time scale $t_{\rm mix}$ ({\it mixing time}), then decays to zero exponentially with a time constant $t_{\rm rel}$ ({\it relaxation time}) with $\lim_{N\to\infty} t_{\rm rel} / t_{\rm mix} = 0$.
    (b) Even though the relaxation time itself may be growing with $N$, upon rescaling time by $t\mapsto t/t_{\rm mix}$ the approach to equilibrium becomes step-like as $N\to\infty$. In this sense the dynamics has a sharp equilibration transition.     \label{fig:cutoff} 
    }
\end{figure}

A key technical step in this work will be to map the evolution of quantum distinguishability measures, Eq.~\eqref{eq:pe_criterion_tracedistance}, to the equilibration of classical Markov chains over finite sets. For this purpose it is useful to briefly review some notations and facts about this subject. 

A Markov chain on a finite set can be simply defined by a stochastic matrix $\Gamma_{ij}$, with $i,j$ ranging over the set. `Stochastic' means that $\Gamma_{ij} \geq 0$ for all $i,j$, with $\sum_i \Gamma_{ij} = 1$ for all $j$. $\Gamma_{ij}$ represents the probability of a transition from $j$ to $i$. A state of the Markov chain is given by a probability distribution $p(i)$. In each time step, $p$ is updated according to $p(i) \mapsto \sum_j \Gamma_{ij} p(j)$. We say that the Markov chain equilibrates if $\Gamma^t p$ converges to a late-time steady state $\pi$ as $t\to\infty$ (in this work we will always deal with Markov chains that admit a unique steady state). 

The equilibration of finite-sized Markov chains is characterized by two fundamental time scales~\cite{levin_markov_2017}. The first is the {\it mixing time}, 
\begin{equation}
    t_{\rm mix}(\epsilon) = \min\{t:\, D_{\rm tv} (p_t,\pi) \leq \epsilon\},
\end{equation}
describing the approach to the steady distribution. Here $D_{\rm tv}$ is the total variation distance, $D_{\rm tv}(p,q) = \frac{1}{2} \sum_i |p(i)-q(i)|$, and one should maximize over all possible initial states.
The second important time scale is the {\it relaxation time}, defined as the inverse spectral gap of the transition matrix $\Gamma_{ij}$: 
\begin{equation} 
t_{\rm rel} = \frac{1}{1-\lambda_1},
\end{equation}
where the spectrum of $\Gamma_{ij}$ is given by $1 = \lambda_0 > \lambda_1 \geq \lambda_2 \geq \dots > -1$. The relaxation time describes the asymptotic late-time decay towards the steady state. 

In many cases of interest in physics, these two time scales have the same behavior: e.g., for the simple random walk on a line of length $L$, both time scales are $\propto L^2$. This gives a scale-invariant behavior, with the approach to equilibrium depending on the single ratio $t / L^2$. 
However, in some cases (including random walks on expander graphs~\cite{hoory_expander_2006}, which are highly connected spaces) the two time scales can be different, with $t_{\rm rel} / t_{\rm mix} \to 0$ in the large-system limit. Such a separation of time scales, sketched in Fig.~\ref{fig:cutoff}, is conjectured~\cite{peres_conjecture,chen_cutoff_2008} to give rise to the {\it cut-off phenomenon}~\cite{diaconis_generating_1981,diaconis_cutoff_1996,chen_cutoff_2008,lacoin_cutoff_2011,chen_cutoff_2021}: an abrupt transition to equilibrium, defined as
\begin{align} 
 \lim_{N\to\infty} D_{\rm tv} ( p_{t_{\rm mix}(N) + \theta t_{\rm rel}(N)}, \pi)  & = \chi(\theta).
\end{align} 
Here $\chi(\theta)$ is an $N$-independent function that goes to 1 as $\theta\to-\infty$ and to 0 as $\theta \to +\infty$. 
Informally, as $N\to \infty$, the total variation distance from equilibrium exhibits a scaling collapse on the form $\chi[(t-t_{\rm mix})/t_{\rm rel}]$, where crucially the mid-point of the collapse ($t_{\rm mix}$) and its width ($t_{\rm rel}$) scale differently in the large-system limit, $t_{\rm rel} = o(t_{\rm mix})$.
Thus if we express time in units of $t_{\rm mix}$, the distance from equilibrium becomes {\it step-like} in the infinite-system limit, going from almost-maximal to almost-zero over a vanishingly narrow window of time, $1 \pm t_{\rm rel} / t_{\rm mix}$. This is sketched in Fig.~\ref{fig:cutoff}(b). 

\subsection{Random walks on permutation groups}\label{sec:review_ip_rw}

In this work we will be especially interested in Markov chains over permutation groups and related processes. Here we give a brief review of some essential concepts; for a more thorough discussion, see e.g. Ref.~\cite{levin_markov_2017}. 

Consider a graph $(V,E)$ with vertices $V = \{v_i\}$ and edges $E = \{(v_i,v_j)\}$. The {\it interchange process} (\IP) is a Markov chain over the permutation group $\perm{|V|}$ whose unit step consists of drawing an edge $(v_i,v_j)\in E$ uniformly at random and transposing its vertices $v_i$, $v_j$. The composition of $t$ such transpositions $\tau_{i,j}$ produces a random element $\sigma_t = \tau_{i_t,j_t} \circ \cdots \circ \tau_{i_1,j_1}$ of the permutation group $\perm{|V|}$. The \IP\ is the Markov chain $\{\sigma_t\}_{t=0}^\infty$ on $\perm{|V|}$.

The \IP\ induces other Markov chains of interest. 
The {\it random walk} (\RW) with initial state $v_i$ is the Markov chain $\{\sigma_t(v_i)\}_{t=0}^\infty$. This has the intuitive meaning we expect for a random walk on the graph $(V,E)$: a particle initialized at vertex $v_i$ can hop with equal probability to any of its neighboring vertices (or stay in place). 
Additionally, the {\it exclusion process} with $m$ particles (\EX{m}) is defined as the Markov chain $\{\sigma_t(A)\}_{t=0}^\infty$ where $A\subset V$ is a collection of $m$ distinct vertices, and $\sigma_t(A) \equiv \{\sigma_t(v_i): v_i\in A\}$ is its image under permutation $\sigma_t$. This process describes the evolution of $m$ indistinguishable, hard-core particles on the graph, where at each time step the occupation number of two neighboring vertices get swapped.
When $m=1$, this clearly reduces to the random walk: $\EX{1} = \RW$. 

These three Markov chains ($\RW$, $\EX{m}$ and $\IP$) are connected by a remarkable result known as {\it Aldous' spectral gap conjecture}\footnote{Despite the name, this is a proven theorem\cite{caputo_proof_2010}.}, which states that 
\begin{equation}
    t_{\rm rel}^{\IP} = t_{\rm rel}^{\RW} = t_{\rm rel}^{\EX{m}} \quad \forall\, m. \label{eq:aldous}
\end{equation}
The inequalities $t_{\rm rel}^{\IP} \geq t_{\rm rel}^{\EX{m}} \geq t_{\rm rel}^{\RW}$ are a simple consequence of the fact that \EX{m}\ is a {\it sub-process} of \IP~\cite{caputo_proof_2010} (informally, a projection of \IP\ onto a smaller state space), and that $\RW$ is a subprocess of $\EX{m}$ ($m>1$). 
However, the other direction of the inequality, $t_{\rm rel}^{\IP} \leq t_{\rm rel}^{\EX{m}} \leq t_{\rm rel}^{\RW}$, is specific to \IP\ and is highly nontrivial---it was conjectured in 1992 and proved in 2009~\cite{caputo_proof_2010}.

Beyond \IP, in this work we will consider other Markov chains over permutation groups. In particular, we will consider a setting where the unit time step does not just transpose two neighboring vertices, like in \IP, but rather performs one of a more general set of ``elementary permutations'' in $\perm{|V|}$ sampled randomly at each time step. 
Such a process also gives a Markov chain over $\perm{|V|}$ and also induces ``generalized exclusion processes'' through the action of $\perm{|V|}$ on the subset spaces $\Sigma_m$, in the same way as \IP\ induces \EX{m}. 
These ``generalized exclusion processes'' differ from the standard one as they can allow multiple particles to hop in a correlated way in each time step. Aldous' spectral gap conjecture, Eq.~\eqref{eq:aldous}, does not apply to these more general cases. In Sec.~\ref{sec:local-circuit} we will in fact encounter a case where the $m$-particle ``generalized exclusion process'' relaxes more slowly than the single-particle random walk.


\section{Ensemble pseudothermalization in automaton circuits}
\label{sec:automaton}

\begin{figure}
    \centering
    \includegraphics[width=0.6\columnwidth]{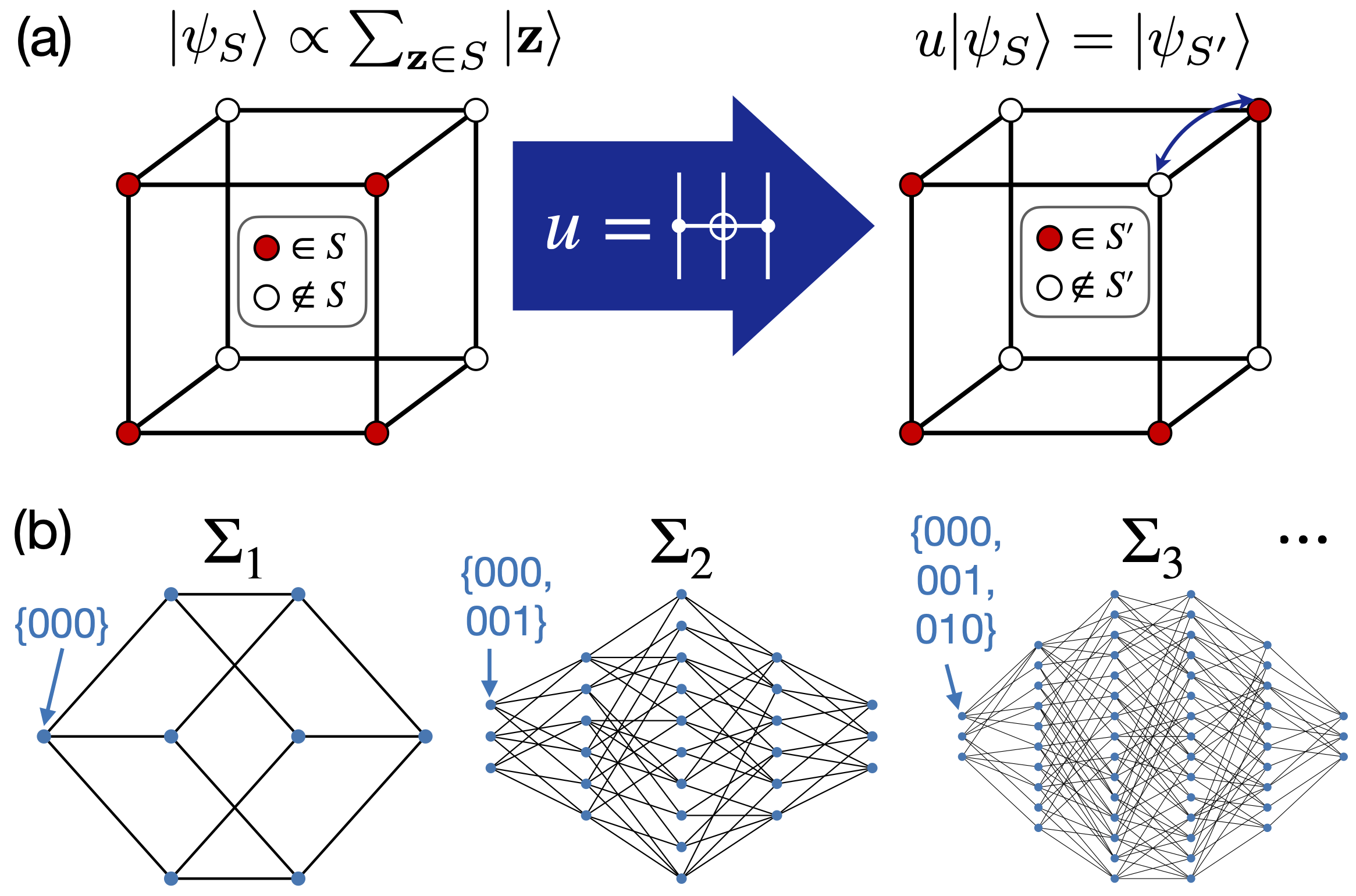}
    \caption{(a) Subset states and their transformation under automaton gates, illustrated here for $N=3$ qubits and a $\textsf{CCX}$ gate. 
    (b) Spaces of subsets of the computational basis, $\Sigma_m = \{S\subset \mathbb{Z}_2^N:\, |S|=m\}$, shown for $N=3$ and $m=1,2,3$. Vertices represent subsets $S\in\Sigma_m$, edges connect subsets $S,S'$ that are mapped to each other under an element of the automaton gate set $\mathcal{G}$ (self-edges are omitted for simplicity).}
    \label{fig:automaton}
\end{figure}

\subsection{Criteria for ensemble pseudothermalization}

We aim to introduce models of quantum dynamics that achieve pseudoentanglement as their equilibrium state. However, fully-generic quantum dynamics tends to maximize actual entanglement. The easiest way to avoid this fate is to constrain the dynamics to preserve the `subset(-phase) state' property, Sec.~\ref{sec:review_pe}. 
Reversible classical circuits, also known as {\it automaton} circuits in the quantum context\footnote{Note that some works use a slightly different definition of automaton gates which also includes phases~\cite{iaconis_quantum_2021}.}~\cite{gopalakrishnan_operator_2018,iaconis_anomalous_2019,iaconis_quantum_2021,klobas_exact_2021}, have this property:
an automaton circuit $U_t$ corresponds to a permutation $\sigma_t \in \perm{D}$ of the computational basis, 
$U_t \ket{\mathbf z} \equiv \ket{\sigma_t(\mathbf z)}$; thus, its action on a subset-phase state $\ket{\psi_{S_0,f_0}}$ is 
\begin{align}
    U_t \ket{\psi_{S_0,f_0}} 
    & = \frac{1}{\sqrt K} \sum_{\mathbf z \in S_0} (-1)^{f_0(\mathbf z)} \ket{\sigma_t(\mathbf z)}
    = \ket{\psi_{S_t,f_t}}, \label{eq:automaton_action}
\end{align}
where $S_t = \{\sigma_t(\mathbf z):\, \mathbf z \in S_0\} \equiv \sigma_t(S_0)$ is another subset of cardinality $K$, and $f_t \equiv f_0 \circ \sigma_t^{-1}$ is another phase function. 
In particular, if $f_0$ is uniformly random then so is $f_t$, and if $f_0$ is constant then so is $f_t$. It follows that automaton dynamics preserves the `subset(-phase) state' property. Further, it preserves the cardinality of the subsets, $K$. 
The action of automaton gates on subsets is sketched in Fig.~\ref{fig:automaton}(a).
An immediate consequence of Eq.~\eqref{eq:automaton_action} is that, under arbitrary automaton dynamics, the entanglement entropy about any cut cannot grow beyond $\log K$ (as $K$ upper-bounds the Schmidt rank of the state). The maximum amount of entanglement is thus controlled by the choice of initial state, and saturation to the Page curve can be avoided.

In this work we consider the following type of dynamics.
A distribution of initial states $\ket{\psi_{S_0,f_0}}$ is specified (these may be subset states or subset-phase states depending on the allowed values of $f_0$);
then, each step of the time evolution consists of sampling an automaton gate $u_i$ uniformly at random from a fixed gate set $\mathcal G$ and applying it to the state. (We will discuss requirements on the gate set $\mathcal G$ in the following.) 
After $t$ time steps, this results in an ensemble of time-evolved states $\{U_t \ket{\psi_{S_0,f_0}} = \ket{\psi_{S_t,f_t}} \}$ defined by all possible initial states $\{\ket{\psi_{S_0,f_0}}\}$ and all possible choices of $U_t = u_{i_t} u_{i_{t-1}} \cdots u_{i_1} \in \mathcal{G}^{\circ t}$.
The relevant notion of equilibration for this state ensemble is whether, at late times, it becomes {\it uniformly distributed} on the appropriate (subset or subset-phase) state space. If that is the case, and subject to usual constraints on the subset size $K$ (see Sec.~\ref{sec:review_pe}), then our dynamics achieves pseudoentanglement at late times. This is the condition we termed {\it ensemble pseudothermalization} in Sec.~\ref{sec:intro}: while the upper limit $K$ on the Schmidt rank prevents maximization of the entanglement entropy and equilibration to a Page state, this departure from the predictions of quantum thermalization may be practically undetectable, in a sense that is rigorously formalized through the concept of pseudoentanglement. 

A first question to address is whether such equilibrium is achieved at all for a given automaton gate set $\mathcal G$.
Let us focus on the subsets and neglect the phases for now. 
The action of automaton circuits on subset states, $U_t\ket{\psi_{S_0}} = \ket{\psi_{S_t}}$, defines a Markov chain $\{S_t\}$ over the space of size-$K$ subsets of the computational basis, 
\begin{equation} 
\Sigma_K \equiv \{S \subset \{0,1\}^N: |S|=K\}.
\end{equation}
The transition matrix of the Markov chain, $\Gamma_{S',S}$, is specified by the gate set $\mathcal G$:
\begin{equation}
    \Gamma_{S',S} = \frac{1}{|\mathcal G|} \sum_{u\in\mathcal G} \delta_{S',u(S)}.
\end{equation}
This transition matrix represents the fact that each gate $u\in\mathcal{G}$ is chosen with uniform probability $1/|\mathcal G|$, and maps each subset $S$ to a subset $S'=u(S) \equiv \{u(\mathbf z):\, \mathbf{z}\in S\}$ (possibly the same subset). 
Viewing $\Gamma$ as an adjacency matrix, we can thus picture the $\Sigma_K$ spaces as graphs. Some minimal examples are shown in Fig.~\ref{fig:automaton}(b), and reveal substantial complexity already for very small systems and number of copies considered. 
To achieve pseudoentanglement, the Markov chain $\{S_t\}$ must converge at late times to the uniform distribution on $\Sigma_K$, denoted here as $\pi_K$:
\begin{equation}
    \lim_{t\to\infty} p_t = \pi_K,
\end{equation}
where $p_t(S)$ is the probability distribution over subsets after $t$ time steps. 

To guarantee the uniform distribution $\pi_K$ as a unique steady state, the Markov chain must be~\cite{levin_markov_2017}: 
\begin{enumerate}
    \item[(i)] {\it Reversible}---the Markov chain transition matrix is symmetric, $\Gamma_{S',S} = \Gamma_{S,S'}$ for all $S,S'$. This property is ensured, e.g., by picking a gate set $\mathcal G$ such that $u^2=I$ for all $u\in \mathcal G$. 
    \item[(ii)] {\it Irreducible}---for any two states $S,S'\in\Sigma_K$, there exists a sequence of gates in $\mathcal G$ that maps $S$ to $S'$. 
    Irreducibility will be discussed separately for each gate set $\mathcal{G}$ in the following.
    \item[(iii)] {\it Aperiodic}---the chain can connect any two states $S,S'$ with paths of any (sufficiently large) length\footnote{The prototypical counterexample is given by non-idle walks on bipartite graphs, where all paths connecting $S$ to $S'$ have the same length modulo 2. In that case the populations of the two partitions of the graph can keep oscillating indefinitely preventing convergence to a steady state.}. One way to ensure aperiodicity is to make the chain `idle', i.e., allow it to skip a step with some probability. This can be done by including the identity gate in $\mathcal G$. 
\end{enumerate}

For a gate set $\mathcal G$ with all three properties above, we can conclude that $p_t \to \pi_K$ as $t\to\infty$~\cite{levin_markov_2017}. This suffices to establish pseudoentanglement at late times in two cases:
\begin{itemize}
\item when the initial state distribution is restricted to a constant phase function, e.g. $f_0 = 0$, we obtain the {\it random subset} ensemble~\cite{giurgica-tiron_pseudorandomness_2023,jeronimo_pseudorandom_2024};
\item when the initial state distribution includes a uniformly random phase function $f_0$, we obtain the {\it random subset-phase} ensemble~\cite{aaronson_quantum_2023}. 
\end{itemize} 
Both of these ensembles are pseudoentangled, subject to certain constraints on $K$, as we reviewed in Sec.~\ref{sec:review_pe}.

\subsection{Mapping to classical Markov processes}

Having established that the steady-state ensembles in these two cases are pseudoentangled, it remains to understand how the system approaches this condition over the course of its dynamics, and what are the time scales involved. As already mentioned, we will call this process {ensemble pseudothermalization}, in analogy with how conventional thermalization in quantum many-body systems can be defined from the saturation of entanglement of local subsystems\footnote{Pseudoentanglement however requires an ensemble of states, hence `ensemble' pseudothermalization.}
Two key technical results will allow us to make progress on these questions by connecting pseudoentanglement with the equilibration of certain classical Markov chains.

\begin{proposition}[subset-phase states]
\label{proposition:subsetphase}
Consider an ensemble of subset-phase states $\mathcal{E}_p = \{\ket{\psi_{S,f}}:\ S\sim p,\ f\sim {\rm unif.}\}$, specified by a uniformly random distribution over phase functions and a general probability distribution $p$ over $\Sigma_K$ (subsets). We have, for all integers $1\leq m\leq K$,
\begin{equation}
    D_{\rm tr}\left( \rho^{(m)}_{\mathcal E_p}, \rho^{(m)}_{\rm Haar} \right)
    = D_{\rm tv} \left( \Phi_{m\leftarrow K}[p], \pi_m \right) + O \left(\frac{m^2}{K}\right)
    \label{eq:thm_randomphase}
\end{equation}
Here $D_{\rm tr}$ and $D_{\rm tv}$ are the trace distance between quantum states and total variation distance between probability distributions, respectively; 
$\rho^{(m)}$ denotes the average $m$-copy state of an ensemble, i.e.
$\rho^{(m)}_{\mathcal E_p} = \mathbb{E}_{S\sim p} \mathbb{E}_f[\ketbra{\psi_{S,f}}^{\otimes m}]$ 
and $\rho^{(m)}_{\rm Haar} = \int {\rm d} \psi_{\rm Haar} \ketbra{\psi}^{\otimes m}$; 
and $\Phi_{m\leftarrow K}$ is a stochastic map from distributions on $\Sigma_K$ to distributions on $\Sigma_m$, defined by
\begin{equation}
    \Phi_{m\leftarrow K}[p](S') = \binom{K}{m}^{-1} \sum_{S\in\Sigma_K:\, S'\subset S} p(S).
    \label{eq:phi_channel_def}
\end{equation}
\end{proposition}

\begin{proof}
See Appendix~\ref{app:proof1}.    
\end{proof}

This statement quantitatively links the distinguishability between quantum ensembles [left hand side of Eq.~\eqref{eq:thm_randomphase}] to the (non-)uniformity of a probability distribution over $\Sigma_m$, that is, over subsets of the computational basis of cardinality $m$ [right hand side of Eq.~\eqref{eq:thm_randomphase}]. 
Since we take $m \leq O(\poly{N})$ and $K \geq \omega(\poly{N})$, the error term on the right hand side is negligible in the regime we are interested in.
Note that the ensemble $\mathcal{E}_p$ is determined by a probability distribution $p$ over $\Sigma_K$, but what controls the $m$-copy distinguishability is a different distribution, $\Phi_{m\leftarrow K}[p]$, over $\Sigma_m$. 
We may view $\Sigma_K$ as the state space for $K$ indistinguishable particles on the hypercube $\{0,1\}^N$; then the stochastic map $\Phi_{m\leftarrow K}$ is akin to taking the marginal distribution over $m$-particle states. Particle indistinguishability however requires symmetrizing over all $\binom{K}{m}$ $m$-particle subsets of each $K$-particle subset, which is the content of Eq.~\eqref{eq:phi_channel_def}.

In Appendix~\ref{app:phi_maps} we prove various useful properties of the stochastic maps $\Phi$. Among these is the fact that, if we view $p$ as the distribution on $\Sigma_K$ induced by the action of a random automaton circuit, then $\Phi_{m\leftarrow K}[p]$ is the distribution induced by the same process on $\Sigma_m$, up to an average over initial states. 
Then Eq.~\eqref{eq:thm_randomphase} maps the $m$-copy distinguishability between quantum state ensembles onto the mixing of a classical Markov chain on $\Sigma_m$. 

\begin{proposition}[subset states]
\label{proposition:subset}
Consider an ensemble of subset states $\mathcal{E}'_p = \{\ket{\psi_{S}}:\ S\sim p\}$, specified by a general probability distribution $p$ over $\Sigma_K$ (subsets). We have, for all integers $1\leq m\leq K/2$,
\begin{equation}
    D_{\rm tr} \left( \rho_{\mathcal{E}'_p}^{(m)},  \rho_{\rm Haar}^{(m)} \right)
    = \frac{1}{2} \left\| \mathcal{M}^{(m)}_p \right\|_{\rm tr} + O\left(\frac{m}{\sqrt{K}}\right).
    \label{eq:thm_nophases}
\end{equation}
Here $\mathcal{M}^{(m)}_p$ is a $|\Sigma_m| \times |\Sigma_m|$ matrix defined by
\begin{align}
    (\mathcal{M}^{(m)}_p)_{S',S''} 
    & = \binom{K}{m}^{-1} \binom{K}{m+\delta} f^{(m+\delta)}_p (S'\cup S''), \\
    f^{(m+\delta)}_p & = \Phi_{m+\delta\leftarrow K}[p]-\pi_{m+\delta},
    \label{eq:thm_nophases_matrix}
\end{align}
where we use subsets $S',S''\in \Sigma_m$ to index the matrix elements of $\mathcal{M}^{(m)}_p$, and define $\delta \equiv |S'\setminus S''|$.
\end{proposition}

\begin{proof}
See Appendix~\ref{app:proof2}. 
\end{proof}

This statement again relates distinguishability between quantum state ensembles [left hand side of Eq.~\eqref{eq:thm_nophases}] to the non-uniformity of certain probability distributions over subsets [right hand side of Eq.~\eqref{eq:thm_nophases}]. However, unlike Eq.~\eqref{eq:thm_randomphase}, here we have contributions from a family of distributions $\{ \Phi_{m+\delta \leftarrow K}[p]\}_{\delta = 0}^m$, rather than from the single distribution $\Phi_{m\leftarrow K}[p]$.
If we view $p$ as arising from a Markov chain on $\Sigma_K$ induced by a random automaton circuit dynamics, then the emergence of pseudoentanglement is again related---albeit in a less straightforward way---to the mixing of the associated Markov chains $\Phi_{m+\delta\leftarrow K}[p]$ on $\Sigma_{m+\delta}$, for all $\delta = 0, 1, \dots, m$.

Having reduced the problem of ensemble pseudothermalization to that of equilibration of certain Markov chains [random walks over the subset spaces $\{\Sigma_m\}$, Fig.~\ref{fig:automaton}(b)], we are now in a position to study the problem quantitatively by leveraging known results on these processes, including those reviewed in Sec.~\ref{sec:review_cutoff} and \ref{sec:review_ip_rw}. 
In the following sections we consider several automaton gate sets $\mathcal{G}$ and initial state distributions, and analyze the associated classical Markov chains to characterize the process of ensemble pseudothermalization.


\section{Spreading of Pseudoentanglement}
\label{sec:spreading_pe}

\subsection{Motivation and setup} \label{sec:spreading_pe_setup}

In this section, we study the problem of how pseudoentanglement spreads in space, focusing on the setting of a smaller pseudoentangled system placed in contact with a larger, trivial system. This problem is reminiscent of analogous settings in the dynamics of entanglement: e.g., the growth of entanglement after quantum quenches where thermalization may happen due to the spreading of entangled quasiparticles in space~\cite{calabrese_evolution_2005,calabrese_entanglement_2009}, or the equilibration of quantum spin chains connected to thermal baths at their edges~\cite{Prosen_matrix_2009,znidaric_transport_2011,morningstar_avalanches_2022,sels_bath-induced_2022}.

We consider a situation in which we are given a random subset-phase ensemble $\mathcal{E}_0 = \{ \ket{\psi_{S_0,f_0}}:\ S_0\sim {\rm unif.},\ f_0\sim {\rm unif.}\}$ on a system $A$ of $N_A$ qubits. 
We take $|S_0| = K$ with $\omega(\poly{N_A}) < K < 2^{o(N_A)}$, so that the initial ensemble on $A$ is pseudoentangled. 
We then introduce a system $\bar{A}$ of $N-N_A$ ancillas initialized in the product state $\ket{0}^{\otimes N-N_A}$, and form a larger system $A\bar{A}$ of $N$ qubits. 
This overall system is still in a subset-phase state, though not a uniformly-random one; the initial state indeed is of the form $\ket{0}^{\otimes N-N_A} \otimes  \ket{\psi_{S_0,f_0}} $. 
We can straightforwardly redefine $S_0$ as a subset of the computational basis of the extended system $\{0,1\}^N$ (rather than the original system $\{0,1\}^{N_A}$), by padding each bitstring with $N-N_A$ leading zeros: $\mathbf{z} = (z_1,\dots z_{N_A}) \mapsto (0,\dots,0,z_1,\dots z_{N_A})$. The phase function $f_0$ is never evaluated on any bitstrings outside the original ones, so it can be trivially extended to a uniformly random binary function on $\{0,1\}^N$. 
Thus after including the ancillas, we have a subset-phase ensemble with uniformly random phases, but highly-correlated subsets. These states are manifestly not pseudoentangled---one can efficiently distinguish them from Haar-random states by, e.g., measuring $\langle Z_j \rangle$ for any $j > N_A$. 

We then place the original system and ancillas into contact by allowing them to evolve together under automaton circuit dynamics, as discussed in Sec.~\ref{sec:automaton}. The question we aim to address is if and how pseudoentanglement is achieved on the entire system of $N$ qubits over the course of the evolution. 
As the phases are already randomized, it is sufficient to focus on the subsets. 

A first apparent constraint is that the value of $K$ is conserved, even as we add the ancillas: even assuming that the subsets indeed equilibrate to the uniform distribution, pseudoentanglement requires $\omega(\poly{N}) < K < 2^{o(N)}$ (see Sec.~\ref{sec:review_pe}). Combining these constraints with the analogous ones from the initial state ($N \mapsto N_A$) yields the overall constraint
\begin{equation}
    \omega(\poly{N}) < K < 2^{o(N_A)}. \label{eq:constraint_K_NA_N}
\end{equation}
In particular this implies that ${N_A}$ must grow faster than $\log(N)$. Intuitively this means that the number of ancillas added to the original system cannot be too large, otherwise one would ``dilute'' the initial pseudoentanglement too much.

Assuming Eq.~\eqref{eq:constraint_K_NA_N}, it remains to analyze the equilibration of subsets toward the uniform distribution under the chosen model of automaton dynamics. In the rest this section, we discuss how to achieve this under two automaton gate sets $\mathcal{G}$, one made of all-to-all gates and the other of local gates.

\subsection{All-to-all circuit model} 
\label{sec:non-local_gate}

\begin{figure}
    \centering
    \includegraphics[width=0.99\columnwidth]{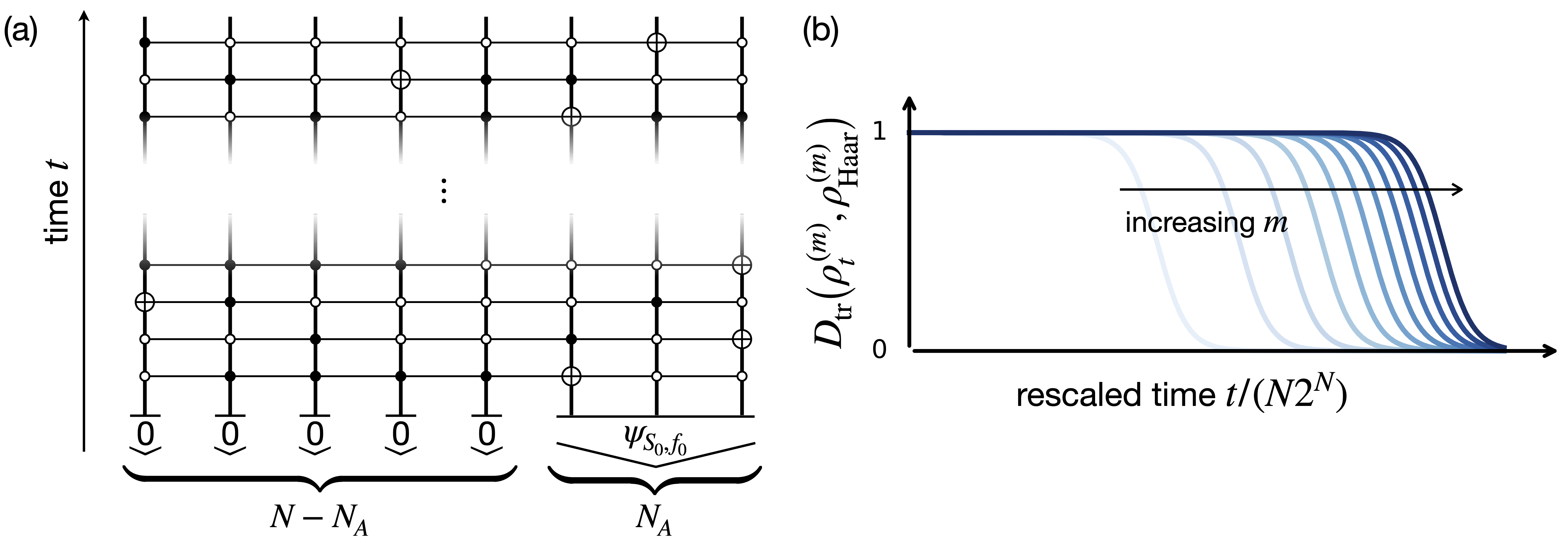}
    \caption{Spreading of pseudoentanglement in an all-to-all automaton circuit model. 
    (a) Sketch of the circuit, showing a subset-phase initial state $\ket{\psi_{S_0,f_0}}$ on subsystem $A$ along with $\ket{0}^{\otimes N-N_A}$ on the rest of the system, and a sequence of random globally-controlled NOT gates, $\CNX{\mathbf a}{i}$. Note that the controls come in two variations, open and solid circles, representing control on 0 and 1 respectively. 
    (b) Sketch of the exact result Eq.~\eqref{eq:alltoall_pseudotherm_result} for the distinguishability transitions in this model. Consecutive indistinguishability times $\pseudothermtime{m}$, $\pseudothermtime{m+1}$ are not sharply separated. 
    }
    \label{fig:nonlocal}
\end{figure}

We start by analyzing a model with a high degree of mathematical tractability, that comes at the cost of all-to-all interactions and exponentially-long equilibration time scales. We emphasize that these latter aspects are not at all necessary to ensemble pseudothermalization in general, and are in fact absent from another model that we study next, in Sec.~\ref{sec:local-circuit}.

We consider a gate set made of {\it globally-controlled} NOT gates: $\mathcal{G} = \{ \CNX{\mathbf a}{i} \}$, where $\CNX{\mathbf a}{i}$ flips qubit $i$ if and only if the remaining $N-1$ qubits are in computational basis state $\ket{\mathbf a}$. There are $ND/2$ gates in $\mathcal{G}$: $N$ choices for the target $i$ and $2^{N-1} = D/2$ choices for the control string $\mathbf a$ ($D=2^N$ is the Hilbert space dimension). 
An example of the resulting circuit family is illustrated in Fig.~\ref{fig:nonlocal}(a). 

Viewed as a permutation of the hypercube $\mathbb{Z}_2^N$, the gate $\CNX{\mathbf a}{i}$ is simply a transposition of two neighboring vertices, 
\begin{equation}
(\mathbf{a}_{1:i-1},0,\mathbf{a}_{i+1:N}) \xleftrightarrow{\CNX{\mathbf a}{i}} (\mathbf{a}_{1:i-1},1,\mathbf{a}_{i+1:N}),
\end{equation} 
where we used the shorthand $\mathbf{a}_{i:j}= (a_i,\dots a_j)$. 
Clearly by a sequence of such neighbor transpositions one can transpose two arbitrary vertices\footnote{Given a path from $\mathbf x$ to $\mathbf y$ on the hypercube, the sequence of neighbor transpositions along the path and back has the effect of transposing $\mathbf x$ and $\mathbf y$ while leaving all other sites fixed.} $\mathbf x\leftrightarrow\mathbf y$; with these one can generate the whole permutation group $\perm{D}$.  
It follows that $\mathcal{G}$ induces an irreducible Markov chain on $\Sigma_m$ for all values of $m$.
The Markov chains induced by $\mathcal{G}$ on the subset spaces $\{\Sigma_m\}$ are also clearly reversible ($\CNX{\mathbf a}{i}$ is its own inverse) and aperiodic (since $|S| = m\ll 2^N$, it is always possible to draw a control-target pair $(\mathbf{a},i)$ such that $\CNX{\mathbf a}{i}$ acts trivially on all bitstrings in $S$, making the walk idle). 
By the discussion in Sec.~\ref{sec:automaton} we have that the unique steady state of the Markov chain on $\Sigma_K$ is the uniform distribution $\pi_K$. Thus at late times this dynamics produces the random subset-phase state ensemble over the entire system, which in the relevant regime of $(K,N)$ is pseudoentangled. 

Having addressed the $t\to\infty$ limit of the dynamics, it remains to understand how pseudoentanglement is achieved as a function of time. To this end, we can use Proposition~\ref{proposition:subsetphase}, 
which relates the $m$-copy distinguishability between the quantum state ensembles, $\mathcal{E}_{p_t}$ and Haar, to the mixing of a Markov chain $\Phi_{m\leftarrow K}[p_t]$ on $\Sigma_m$, with $p_t$ the probability distribution over subsets at time $t$ during the dynamics. 

We note that the Markov chain on the permutation group $\perm{D}$ induced by this gate set $\mathcal{G}$ is precisely the interchange process (\IP) on the hypercube $\mathbb{Z}_2^N$, reviewed in Sec.~\ref{sec:review_ip_rw}:
each gate picks an edge of the hypercube uniformly at random and transposes its two vertices.
It follows that the Markov chains $\Phi_{m\leftarrow K}[p_t]$ over the spaces of subsets $\Sigma_m$ are exactly the exclusion processes $\EX{m}$. 
These are well-studied problems for which we can rely on powerful known results. 

First of all, Aldous' spectral gap conjecture~\cite{caputo_proof_2010}, Eq.~\eqref{eq:aldous}, states that the relaxation time of $\EX{m}$ is equal, for all $m$, to the relaxation time of the simple random walk (\RW).
The simple random walk on the hypercube can be solved straightforwardly (see Appendix~\ref{app:hypercube_rw}), and one finds\footnote{
The exponential scaling in $N$ is due to the fact that each elementary transposition of hypercube vertices is exponentially unlikely to cause the random walker to hop: only $N$ gates out of $ND/2$ cause a nontrivial hop. The conventional result for a random walker who hops once per unit time is $t_{\rm rel}^{\RW} = N/2$; rescaling this by the average number of gates between hops ($D/2$) yields Eq.~\eqref{eq:relaxation_nonlocal} (see also more detailed derivation in Appendix~\ref{app:hypercube_rw}).
} 
$t_{\rm rel}^{\RW} = ND/4$; therefore 
\begin{equation}
    t_{\rm rel}^{\EX{m}} = t_{\rm rel}^{\RW} = \frac{ND}{4} \quad \forall\ m.
    \label{eq:relaxation_nonlocal}
\end{equation}

Secondly, the mixing time for $\EX{m}$ was recently solved~\cite{hermon_exclusion_2020} and was found to scale as
\begin{equation}
    t_{\rm mix}^{\EX{m}} \sim ND \log(Nm). 
    \label{eq:mixing_nonlocal}
\end{equation}
Due to Proposition~\ref{proposition:subsetphase}, this mixing time is also the time taken to achieve $m$-copy indistinguishability.

Based on Eq.~\eqref{eq:relaxation_nonlocal} and \eqref{eq:mixing_nonlocal}, we see that 
\begin{equation}
t_{\rm mix}^{\EX{m}} / t_{\rm rel}^{\EX{m}} \sim \log(Nm) \xrightarrow{N\to\infty} \infty.
\end{equation}
This asymptotic separation between the scaling of relaxation and mixing times is a necessary condition for a {\it cut-off phenomenon}, reviewed in Sec.~\ref{sec:review_cutoff}. It was conjectured that in sufficiently ``generic'' models, such separation should also be a sufficient condition\footnote{
Note that the condition $t_{\rm rel} = o(t_{\rm mix})$ is not sufficient to get a cutoff phenomenon in total variation distance in general (on the other hand it {\it is} sufficient if the distance from equilibrium is measured with a $L^{p}$ norm, $p>1$~\cite{chen_cutoff_2008}). It is possible to have a so-called \textit{pre-cutoff} where multiple step-like transitions arise in the total variation distance. However, this usually requires some particular structure in the connectivity of the underlying graph.
}. 
This cutoff effect would imply a sequence of dynamical transitions in our dynamics: for any number of copies $m \leq O(\poly{N})$, there is a time scale at which the ensemble of output states becomes indistinguishable from Haar-random at the $m$-copy level; this $m$-copy {\it indistinguishability time} in the present model is given by
\begin{equation}
    \pseudothermtime{m} \sim N2^N[\log(mN) \pm O(1)],
    \label{eq:alltoall_pseudotherm_result}
\end{equation}
with the $\pm$ term denoting the temporal width of the dynamical crossover to equilibrium. This result is sketched in Fig.~\ref{fig:nonlocal}(b).  We note that, while each distinguishability transition is sharp (in the sense of $t_{\rm rel}^{(m)} / t_{\rm mix}^{(m)} \to 0$ at large $N$), consecutive transitions $\pseudothermtime{m}$, $\pseudothermtime{m+1}$ are {\it not} sharply separated. This is due to the slow growth of $t_{\rm mix}^{(m)}$ with $m$, namely the fact that $[t_{\rm mix}^{(m+1)}-t_{\rm mix}^{(m)}]/t_{\rm rel}^{(m)} \sim \log\frac{m+1}{m}$ does not diverge (in fact it goes to zero at large $m$).

This exactly solvable model serves as a nice illustration of the idea of ensemble pseudothermalization, including the possible emergence of a cutoff phenomenon. At the same time, it has several drawbacks: the all-to-all interactions require complex couplings or compilation into deep local circuits; even worse, the resulting time scales for relaxation and mixing are exponentially long in system size. However, these limitations are artifacts of the model, and not intrinsic to the concept of pseudoentanglement or ensemble pseudothermalization. Below we introduce a {\it local} automaton gate set that does not suffer from these issues. 

\subsection{Local circuit model}
\label{sec:local-circuit}

\begin{figure}
    \centering
    \includegraphics[width=0.9\columnwidth]{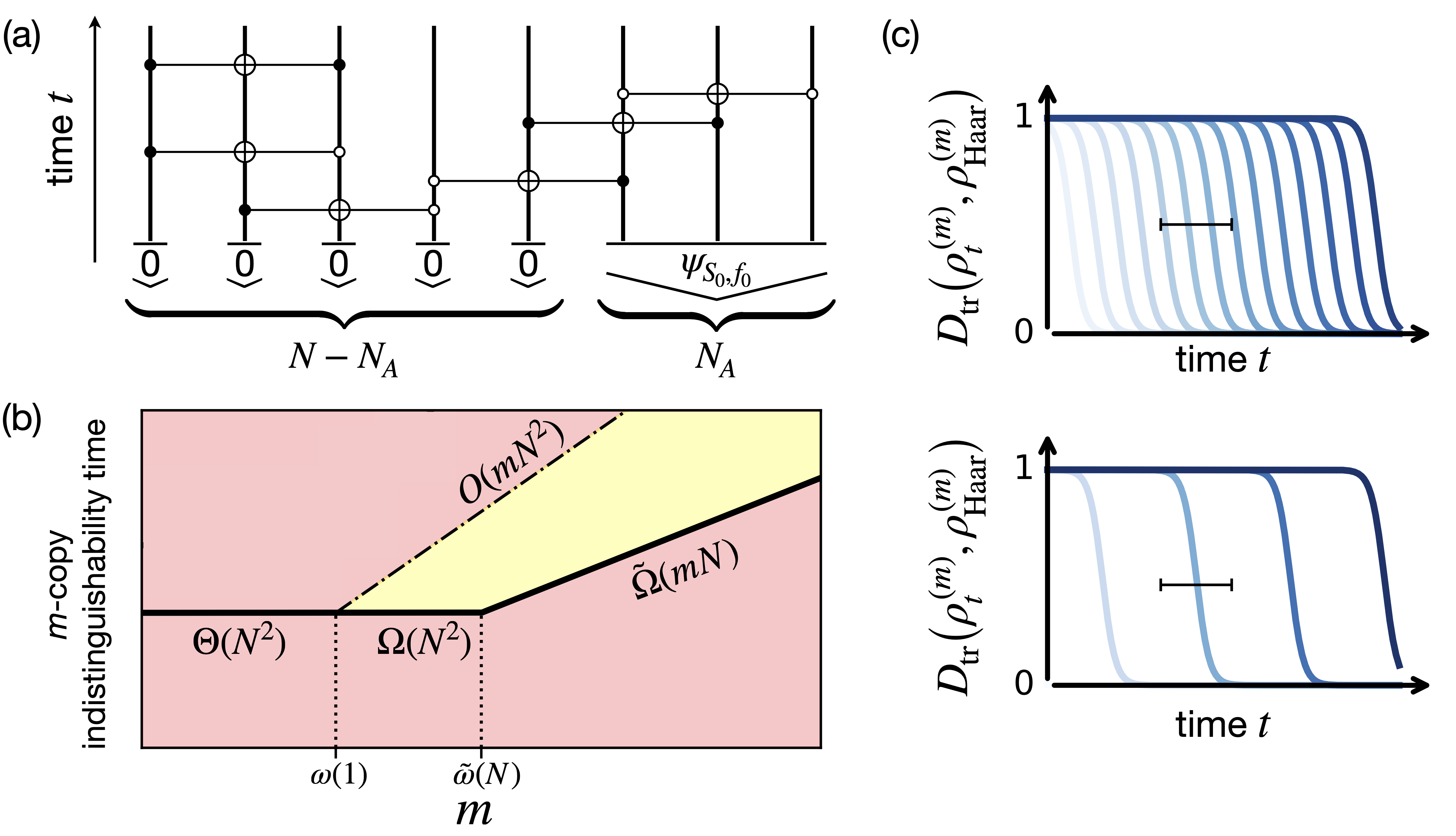}
    \caption{Spreading of pseudoentanglement in a local automaton circuit model. 
    (a) Circuit model, analogous to Fig.~\ref{fig:nonlocal} but with geometrically local $\CCX$ gates.
    (b) Schematic summary of bounds on the indistinguishability time. The $x$ axis qualitatively indicates asymptotic scalings of $m$ vs $N$ as $N\to\infty$. Solid lines denote unconditional lower bounds while the dot-dashed line in an upper bound that depends on Conjecture~\ref{conjecture:relaxation}. Red regions are ruled out while the yellow region is allowed. The lower bounds come from Eq.~\eqref{eq:tmix_counting_lowerbound} (counting) and Eq.~\eqref{eq:tmix_locality_lowerbound} (locality). The upper bound comes from Eq.~\eqref{eq:tmix_bounds_conditional}.
    (c) Consecutive $m$-copy indistinguishability transitions in this model may be overlapping (left) or sharply separated (right) in the thermodynamic limit: answering this question requires tighter bounds on the mixing time.
    }
    \label{fig:local}
\end{figure}

To build a more physically realistic model of dynamics, we seek to restrict the gate set $\mathcal G$ to geometrically local gates on a $d$-dimensional array of qubits. 

The simplest local automaton gate set comprises just local bit flips: $\mathcal{G} = \{X_i\}_{i=1}^N$. It is straightforward to see that the Markov chain induced by $\mathcal{G}$ on $\Sigma_1$ is irreducible (we can map any bitstring to any other via bitflips); however it is already not irreducible on $\Sigma_2$: given a set $S = \{\mathbf z, \mathbf z'\}$, the ``relative coordinate'' $\mathbf z \oplus \mathbf z'$ ($\oplus$ being sum modulo 2) is invariant under bit flips. 
Moving up in complexity, one might try two-qubit automaton gates, generated by NOTs and CNOTs; however, it is possible to show that these gates are also insufficient to achieve irreducibility on $\Sigma_m$ above $m=3$.
Three-qubit automaton gates, on the other hand, can induce irreducible Markov chains on $\Sigma_m$ for all $m \leq D/4$, as we show in Appendix~\ref{app:local_gateset}.

In particular, for the case of a one-dimensional qubit array, we consider the gate set 
$\mathcal{G} = \{u_{iab} = \textsf{C}_{i-1,a} \textsf{C}_{i+1,b} {X}_i\}$
where $a,b\in \{0,1\}$ label a control bitstring and $i \in [N]$ labels the target site (we take periodic boundary conditions, identifying $i = 0$ and $i = N$), as shown in Fig.~\ref{fig:local}(a). 
The $u_{iab}$ gate flips qubit $i$ if and only if its left neighbor $i-1$ is in state $a$ and its right neighbor $i+1$ is in state $b$. 
There are $N$ choices for the target site $i$ and 4 choices for the control string $(a,b)$, giving a total of $4N$ gates in $\mathcal{G}$.
Up to bit flips, each $u_{iab}$ is a Toffoli gate.
It is easy to see that the gate set generates $\CX$ gates, $\textsf{SWAP}$ gates, and bit flip gates as well. In fact it generates arbitrary {\it even} permutations of the computational basis, see Appendix~\ref{app:local_gateset}.

\tocless\subsubsection{First moment}

We start our analysis of this dynamics from the first moment, $m = 1$. The Markov chain on $\Sigma_1 \equiv \mathbb{Z}_2^N$ corresponds once again to the simple random walk on the hypercube, see Appendix~\ref{app:hypercube_rw}. With each gate the walker, a bitstring $\mathbf z$, stays in place with probability 3/4 (if the control string $a,b$ of the randomly-sampled $u_{iab}$ gate does not match $\mathbf z$), otherwise it hops to a neighbor (a bitstring $\mathbf z'$ differing only at one bit) chosen uniformly at random based on the value of $i \in [N]$. 
Using again the known results for the simple random walk on the hypercube, Appendix~\ref{app:hypercube_rw}, we have
\begin{equation}
    t_{\rm rel}^{(1)} = 2N,
    \qquad 
    t_{\rm mix}^{(1)} \sim N\log(N),
    \label{eq:trel_tmix_RW_local}
\end{equation}
with the only difference from the case of Sec.~\ref{sec:non-local_gate} (with $m=1$) being the different average idling time between hops (4 instead of $D/2$). 
Just like in Sec.~\ref{sec:non-local_gate}, the random walk exhibits a cutoff~\cite{diaconis_asymptotic_1990}, with the approach to equilibrium taking place in a relatively narrow time window, $t \sim N[\log N \pm O(1)]$. Intuitively, to achieve equilibrium the walker has to hop in all $N$ directions with high probability (equivalently, all $N$ qubits need to be flipped with high probability). This is an instance of the so-called ``coupon collector's problem'', which gives a scaling $\sim N \log(N)$. 
Note that time $t$ is defined as total number of gates acting on the system; in terms of an intensive time $\tau \equiv t/N$, such that each qubit flips of order once per unit $\tau$, equilibration happens at $\tau \sim \log(N) \pm O(1)$. 

\tocless\subsubsection{Higher moments: Relaxation time}

For higher moments $m>1$, unlike the all-to-all model of Sec.~\ref{sec:non-local_gate}, the Markov chain on $\Sigma_m$ is not exactly solvable (to our knowledge). To make progress we resort to a combination of rigorous bounds and numerical calculations.

\begin{figure}
    \centering
    \includegraphics[width=0.6\columnwidth]{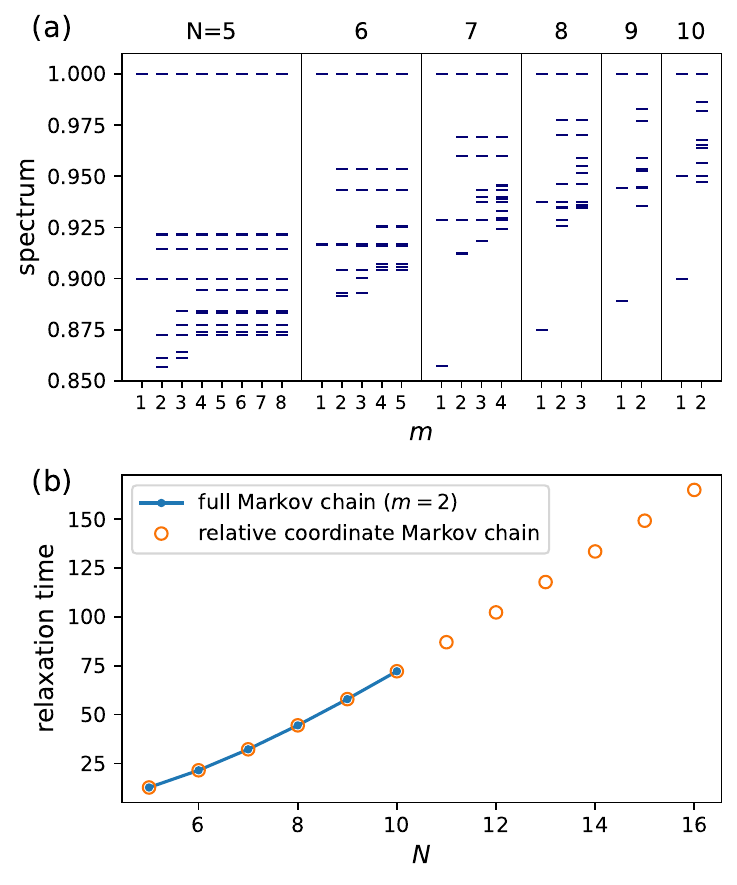}
    \caption{(a) Largest 20 eigenvalues of the Markov chain transition matrix $\Gamma_{S,S'}$ for different system sizes $N=5,\dots 10$ (indicated at the top of each panel) and subset sizes $m$, obtained from exact diagonalization (Lanczos). We consider all pairs $(N,m)$ such that the dimension $|\Sigma_m| = \binom{2^N}{m}$ of the transition matrix is below $2\times 10^7$. For $m=1$ we recover the spectrum of the single-particle random walk, $1-\frac{c}{2N}$, $c \in [N]$; a new leading eigenvalue $\lambda_1 > 1-\frac{1}{2N}$ appears for $m=2$ and persists for all accessible values of $m$. 
    (b) Relaxation time $t_{\rm rel} = 1/(1-\lambda_1)$ as a function of system size $N$ for two Markov chains: the full Markov chain over pairs of bitstrings, $\Gamma$  ($\lambda_1$ taken from (a), $m=2$), and the random walk of the ``relative coordinate'' $\Gamma|_R$ (see main text). The two agree to numerical precision up to the system sizes accessible to both ($N\leq 10$); the data suggests an asymptotically linear scaling of $t_{\rm rel}^{(2)}$ with $N$.}
    \label{fig:spectra}
\end{figure}

We begin by analyzing the relaxation time $t_{\rm rel}^{(m)}$ of the Markov chain on $\Sigma_m$, for $m>1$, through exact numerical diagonalization of the Markov chain transition matrix $\Gamma$ for accessible values of the system size $N$ and subset size $m$. Exact diagonalization is limited by the dimensionality of the problem, given by the size of the subset space\footnote{The effective dimensionality is lower due to the problem's translation and reflection invariance, but only by a factor of $\approx N$.} $|\Sigma_m| = \binom{D}{m} = \binom{2^N}{m}$. This grows explosively with $m$ and $N$, limiting us to $mN\lesssim 30$. 

Results for the spectra of $\Gamma$ for different system sizes $N$ and moments $m$ are shown in Fig.~\ref{fig:spectra}(a). 
Recall that the spectrum of $\Gamma$ has a leading eigenvalue $\lambda_0 = 1$, associated to the uniform distribution $\pi_m$ (its unique steady state), and then subleading eigenvalues $1 > \lambda_1 \geq \lambda_2 \geq \dots$; the inverse spectral gap defines the relaxation time: $t_{\rm rel}^{(m)} = 1/(1-\lambda_1)$. 
For $m=1$ the problem reduces once again to the simple random walk on the hypercube $\mathbb{Z}_2^N$, so the leading eigenvalue $\lambda_0 = 1$ is followed by a $N$-fold degenerate multiplet of eigenvalues $\lambda_1 = 1-\frac{1}{2N}$, giving the relaxation time $t_{\rm rel}^{(1)} = 2N$. 
These eigenvalues appear in the spectra for all values of $m$, where they are associated to the $N$ components of the dipole moment of the $m$-particle distribution on the hypercube, see Appendix~\ref{app:hypercube_rw_multipole}. 

However, for $m\geq 2$, we observe a new second-largest eigenvalue in the spectrum of $\Gamma$, $\lambda_1 > 1-1/2N$, giving longer relaxation times $t_{\rm rel}^{(m)} > t_{\rm rel}^{(1)}$. 
This is unlike the all-to-all model of Sec.~\ref{sec:non-local_gate}, where Aldous' spectral gap conjecture, Eq.~\eqref{eq:aldous} (see Sec.~\ref{sec:review_ip_rw}), forces $t_{\rm rel}^{(m)}$ to be $m$-independent.
Nonetheless, Fig.~\ref{fig:spectra}(a) shows no further $m$-dependence in $\lambda_1$ beyond $m = 2$. 

By numerical inspection, we find that the new slowest mode of $\Gamma$ for $m = 2$, the function $f_1(S)$ such that $\Gamma f_1 = \lambda_1 f_1$, depends on the subset $S = \{\mathbf{z}, \mathbf{z}'\}$ only through the ``relative coordinate'' $\mathbf r \equiv \mathbf{z} \oplus \mathbf{z}'$. 
In Appendix~\ref{app:hypercube_rw_relative} we show that the subspace $R$ spanned by functions that depend on $S$ only through the relative coordinate $\mathbf r$ is invariant under the Markov chain. Thus we can block-diagonalize $\Gamma = \Gamma|_R \oplus \Gamma|_{R^\perp}$, where the block $\Gamma|_R$ describes a random walk of $\mathbf r$. 
Numerical diagonalization confirms that the leading eigenvalue $\lambda_1$ of $\Gamma$ comes from the $\Gamma|_R$ block: results for $\lambda_1$ agree to numerical precision between the two Markov chains. Furthermore, due to the reduced dimensionality ($\sim D$ for $\Gamma|_R$ compared to $\sim D^2$ for the full $\Gamma$), we are able to numerically diagonalize the restricted Markov chain up to larger values of system size $N$. Results, shown in Fig.~\ref{fig:spectra}(b), indicate a clear linear scaling $t_{\rm rel}^{(2)} = \Theta(N)$. 

These observations ($t_{\rm rel}^{(2)}$ scaling linearly in $N$ and $t_{\rm rel}^{(m)}$ being independent of $m \geq 2$ for all numerically accessible values) lead us to formulate the following conjecture:

\begin{conjecture}[scaling of relaxation time]
\label{conjecture:relaxation}
For all $m\geq 2$, the Markov chain on $\Sigma_m$ induced by the local automaton gate set $\mathcal{G}$ has relaxation time $t_{\rm rel}^{(m)} = \Theta(N)$, independent of $m$.
\end{conjecture}

We cannot rule out the possibility of further changes in the relaxation times $t_{\rm rel}^{(m)}$ at larger values of $m$ beyond those accessible in numerics. In Appendix~\ref{app:local_relaxation} we show an unconditional upper bound 
\begin{equation}
    t_{\rm rel}^{(m)} \leq O(mN^5), \label{eq:trel_unconditional_upperbound}
\end{equation}
which can be tightened to $t_{\rm rel}^{(m)} \leq \tilde{O}(mN)$ using recent results by He and O'Donnell~\cite{he_pseudorandom_2024}.
Yet more recent results by Chen {\it et al.}~\cite{chen_incompressibility_2024} rigorously prove a result equivalent to our Conjecture~\ref{conjecture:relaxation} in models of 3-local, but not geometrically local, random reversible circuits. Adapting their result to the present setting, and thus proving our Conjecture~\ref{conjecture:relaxation}, is as an interesting direction for future work.

\tocless\subsubsection{Higher moments: Mixing time}
\label{sec:local_gate_mixing}

\begin{figure}
    \centering
    \includegraphics[width=0.9\columnwidth]{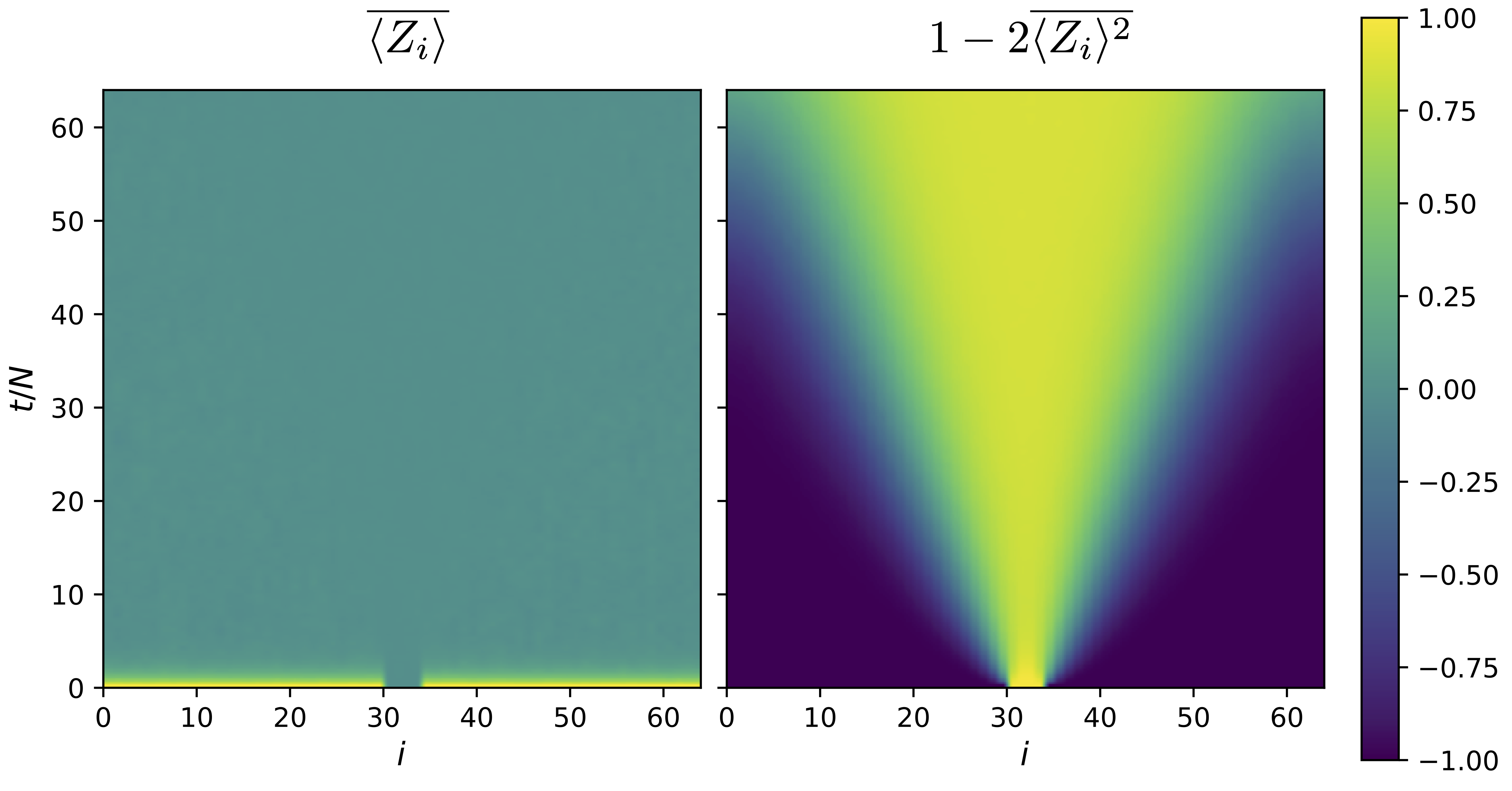}
    \caption{Numerical simulation of local random automaton circuits on a chain of $N = 64$ qubits, initialized in the product state $\ket{0}^{\otimes N-N_A} \otimes \ket{+}^{\otimes N_A}$ with $N_A = 4$ (we shift the $\ket{+}$ qubits to the center of the plot for clarity). Data averaged over $10^4$ random circuit realizations.
    Left: local expectation values averaged over circuit realizations, $\overline{\langle Z_i\rangle}$, can relax under local bit flips and thus equilibrate quickly. 
    Right: higher-moment quantities, such as $\overline{\langle Z_i\rangle^2}$ (related to the variance of $\langle Z_i\rangle$ across circuit realizations), require the propagation of coherence across the system. Their equilibration is thus constrained by a light cone, $t/N \geq \Omega(N)$. Note $t$ is the total number of applied gates.}
    \label{fig:lightcone}
\end{figure}

Assuming the validity of Conjecture~\ref{conjecture:relaxation}, $t_{\rm rel}^{(m)} = \Theta(N)$, we move on to the question of mixing times. A general two-sided bound on the mixing time in terms of the relaxation time is given by~\cite{jerison_general_2013,levin_markov_2017}
\begin{equation}
    (t^{(m)}_{\rm rel}-1)\log(\frac{1}{2\epsilon}) \leq t_{\rm mix}^{(m)}(\epsilon) \leq t_{\rm rel}^{(m)} \log (|\Sigma_m| / \epsilon), \label{eq:mix_eigenbound}
\end{equation}
where $\epsilon$ is the accuracy threshold on total variation distance used to define mixing; see Sec.~\ref{sec:review_cutoff}. We take $\epsilon = 1/4$ in the following. 
Plugging Conjecture~\ref{conjecture:relaxation} into Eq.~\eqref{eq:mix_eigenbound} yields the asymptotic bounds
\begin{equation}
    \Omega(N) \leq t_{\rm mix}^{(m)} \leq O(mN^2),
    \label{eq:tmix_bounds_conditional}
\end{equation}
where we use $\log |\Sigma_m| = \log \binom{D}{m} \leq mN$.

The lower bound Eq.~\eqref{eq:mix_eigenbound} clearly cannot be used to gap $t_{\rm mix}^{(m)}$ away from $t_{\rm rel}^{(m)}$ (necessary condition for a cut-off phenomenon). Therefore, to study the possible presence of a cut-off we must find stronger lower bounds on the mixing time. 
We show two such bounds next.

\begin{fact}[counting bound] \label{fact:countingbound}
    For all $m$, the mixing time of the Markov chain induced by the local gate set $\mathcal{G}$ obeys
    \begin{equation}
        t_{\rm mix}^{(m)} \geq \Omega(mN/\log N). \label{eq:tmix_counting_lowerbound}
    \end{equation}
\end{fact}

\begin{proof}
    Let us count the number of states that can be reached from an initial state $S_0 \in \Sigma_m$ in $t$ time steps. As each state is connected to at most $|\mathcal G| = 4N$ different states, the volume of a ball of radius $t$ in $\Sigma_m$ (viewed here as a graph with edges given by the action of $\mathcal{G}$) is at most $(4N)^t$. To achieve mixing, the distribution $p_t$ needs to be supported on, say, at least $1/2$ of the elements of $\Sigma_m$. Thus we must have $(4N)^t \geq |\Sigma_m| / 2$, giving $t\log(4N) \geq \log \binom{D}{m} - \log 2\geq \Omega(mN)$.
\end{proof}

\begin{fact}[locality bound] \label{fact:localitybound}
    For all $m\geq 2$, the mixing time of the Markov chain induced by the local gate set $\mathcal{G}$ on a 1D qubit chain obeys
    \begin{equation}
        t_{\rm mix}^{(m)} \geq \Omega(N^2). \label{eq:tmix_locality_lowerbound}
    \end{equation}
\end{fact}

\begin{proof}
It is easy to see that, starting from a subset $S_0$ where all bitstrings agree outside the subsystem $A$ ($z_j = 0$ for all sites $j=1, \dots N-N_A$, for all bitstrings $\mathbf z\in S_0$), the region outside a ``light cone'' emanating from $A$ remains in a computational basis state: local automaton gates cannot create superpositions from such states, so coherence must spread outward from $A$ one gate at a time. 
As the gate sequence is random, the light cone fluctuates instance by instance, however its width $\ell(t)$ is bounded above\footnote{
We can bound $\ell(t)$ above by noting that $\ell(t+1) = \ell(t) + 1$ with probabiliy at most $2/N$ (the $\CCX$ gate must have the target qubit $i$ at either light cone front), and $\ell(t+1)\leq \ell(t)$ otherwise. Then $\overline{\delta\ell} \leq 2/N$ and thus $\ell(t) \leq \ell(0) + 2t/N$ with finite probability. This gives an upper bound on the butterfly velocity $v \leq 1$. 
} 
by $\ell(t) \leq  N_A+2vt/N$, for some constant `butterfly velocity' $v$, with finite probability. 
Whenever $\ell(t) < N$, the distribution $p_t(S)$ is supported only on a small fraction of the state space $\Sigma_m$ (all subsets $S$ that correspond to a product state outside the light cone), and thus must be far from uniform\footnote{
Another way to argue this is by the distinguishing statistics method in Ref.~\cite{levin_markov_2017} applied to, e.g., $\overline{\langle Z_i\rangle^2}$ for a qubit $i$ far from $A$.
}. Equilibrium can only be reached after the light cone has spread across the whole system, giving $\ell(t) \geq \Omega(N)$ or $t\geq \Omega(N^2)$. 
\end{proof}

We note that, while the counting bound (Eq~\ref{eq:tmix_counting_lowerbound}) applies to all $m$, the locality bound (Eq.~\ref{eq:tmix_locality_lowerbound}) does not apply to $m = 1$; indeed we have $t_{\rm mix}^{(1)} \sim N \log N$, Eq.~\eqref{eq:trel_tmix_RW_local}. This is because ``subsets'' with $m = 1$ are just individual bitstrings, so equilibration does not require coherence (i.e. superpositions) to spread across the system. This physical difference in the effects of locality is illustrated in Fig.~\ref{fig:lightcone}. 
We show results of numerical simulations for circuit-averaged quantities $\overline{\langle Z_i\rangle}$ (a first-moment quantity that can equilibrate under local bitflips) in Fig.~\ref{fig:lightcone}(a) and $\overline{\langle Z_i \rangle^2}$ (a second-moment quantity that requires superpositions to equilibrate) in Fig.~\ref{fig:lightcone}(b). The overline denotes averaging over instances of the random circuit.
``Time'' $t$ here is measured as total number of gates in the circuit, so $t\sim N^2$ here denotes ballistic, not diffusive, spreading of correlations. 

Putting Facts~\ref{fact:countingbound} and \ref{fact:localitybound} together with the upper bound Eq.~\eqref{eq:tmix_bounds_conditional}, we arrive at\footnote{
Note that the mixing time $t_{\rm mix}^{(m)}$ is conventionally based on a worst-case {\it pure} initial state, $p_0(S) = \delta_{S,S_0}$; in the present context, the initial state under consideration is already partially mixed, $p_0(S) = {\rm const.}>0$ if $S$ is made of bitstrings that are identically 0 outside subsystem $A$, $p_0(S) = 0$ otherwise (describing the initial random subset-phase ensemble on $A$, see also Appendix~\ref{app:phi_maps}). 
The upper bound $\pseudothermtime{m} \leq O(mN^2)$ holds {\it a fortiori} for our mixed initial state; the locality lower bound $\pseudothermtime{m} \geq \Omega(N^2)$ and the counting lower bound $\pseudothermtime{m} \geq \Omega(mN/\log N)$ also apply without changes as long as $N_A = o(N)$.
}
\begin{equation}
    \Omega(N^2,mN/\log(N)) \leq t_{\rm mix}^{(m)} \leq O(mN^2).
    \label{eq:localmodel_bounds}
\end{equation}
These bounds are summarized in Fig.~\ref{fig:local}(b). 
We remark again that the upper bound Eq.~\eqref{eq:tmix_bounds_conditional} assumes Conjecture~\ref{conjecture:relaxation}; an unconditional bound $t_{\rm mix}^{(m)} \leq O(m^2 N^6)$ can be derived based on Eq.~\eqref{eq:trel_unconditional_upperbound} (see Appendix~\ref{app:local_relaxation}) and can be further tightened to $t_{\rm mix}^{(m)}\leq \tilde{O}(m^2N^2)$ using recent results~\cite{he_pseudorandom_2024}.  

Several interesting consequences follow form Eq.~\eqref{eq:localmodel_bounds}.
First of all, independently of $m$, we have
\begin{equation}
    \lim_{N\to\infty} t_{\rm mix}^{(m)} / t_{\rm rel}^{(m)} = \infty, \label{eq:local_divergentratio}
\end{equation}
which is a necessary condition for a cut-off, see Sec.~\ref{sec:review_cutoff}.
We thus conjecture the presence of asymptotically sharp distinguishability transitions in this model. 
While numerical simulations are strongly limited by the explosive growth of the state space $|\Sigma_m| = \binom{2^N}{m}$, results of exact simulations of $N = 5$ qubits, shown in Fig.~\ref{fig:mixing_vs_time}, are consistent with the emergence of a cutoff---we clearly see the development of a plateau, followed by a decay $\sim \lambda_1^{t-\delta t}$ of the total variation distance, with a temporal offset $\delta t\sim m$ in the accessible range $1\leq m \leq 8$ (for $N=5$, $m=8$ we have $|\Sigma_8| = \binom{32}{8} \gtrsim 10^7$ states).

\begin{figure}
    \centering
    \includegraphics[width=0.5\columnwidth]{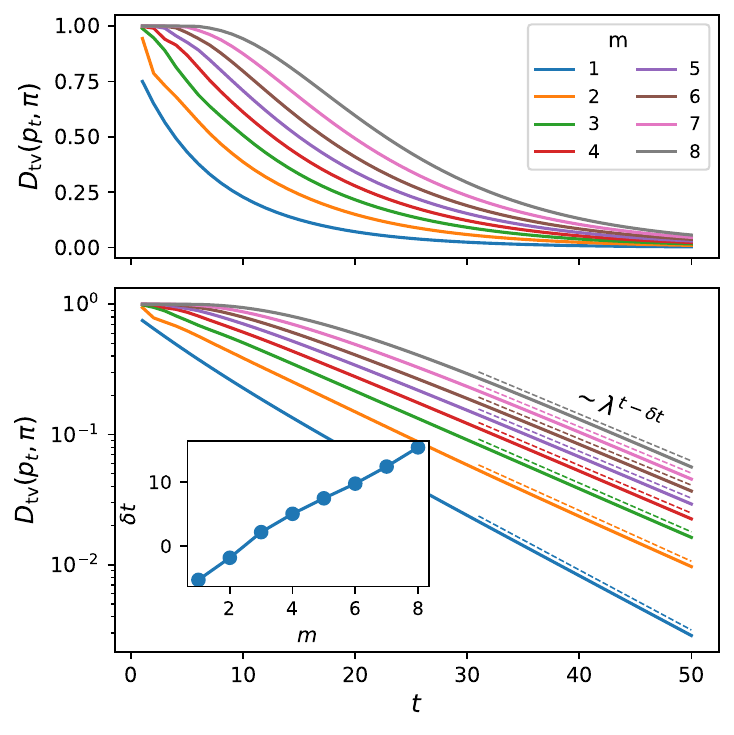}
    \caption{Exact numerical simulation of the Markov chain induced by the local automaton gate set (Sec.~\ref{sec:local-circuit}) on $N=5$ qubits for subset sizes $m = 1,\dots 8$. The initial quantum state is taken to be $\ket{00+++}$, a subset state with $K=8$. 
    The total variation distance from uniformity decays as $D_{\rm tv}(p_t, \pi_m) \sim \lambda^{t-\delta t}$ at late times (bottom panel, dashed lines). Fit results for $\lambda$ (not shown) are consistent with the results of Fig.~\ref{fig:spectra}, with $\lambda \simeq 0.90$ for $m=1$ and $\lambda\simeq 0.92$ for all $m\geq 2$; the temporal offset $\delta t$, a proxy for the mixing time (up to additive constants), is found to scale linearly in $m$ (bottom panel, inset).}
    \label{fig:mixing_vs_time}
\end{figure}

We can then break down the behavior of different moments $m$, based on their scaling with $N$, to obtain an interesting (though still incomplete) picture of the dynamics:
\begin{enumerate}[label = (\roman*)]
    \item For constant $m$, we have $t_{\rm mix}^{(m)} = \Theta(N^2)$, and thus the indistinguishability time
        \begin{equation}
            \pseudothermtime{m} \sim N^2 [ 1 \pm O(1/N)],
        \end{equation}
    with the $\pm$ term denoting the width of the time window in which the transition to equilibrium happens. 
    This suggests a picture based on a `pseudoentanglement front' spreading ballistically through the system, with fast equilibration behind the front.
    In this picture, $m$-copy indistinguishability arises suddenly, and simultaneously for all constant $m$, once the front has swept across the system. 
    \item In the intermediate range $\omega(1) \leq m\leq \tilde{O}(N)$, our bounds are insufficient to pin down the scaling of $\pseudothermtime{m}$; we have
        \begin{equation}
            \pseudothermtime{m} \sim N^2[f(m)\pm O(1/N)],
        \end{equation}
    for some function $\Omega(1)\leq f(m)\leq O(m)$. Depending on the scaling of $f(m)$, it is possible that the `pseudoentanglement front' picture ceases to apply, with a parametrically slower equilibration behind the ballistic front and a separation of consecutive indistinguishability transitions $\pseudothermtime{m}$, $\pseudothermtime{m+1}$. 
    \item At even higher moments $\tilde{\omega}(N)\leq m \leq O(\poly{N})$ (super-linear, but still polynomial in $N$), the counting bound becomes dominant over the locality bound, so the `front' picture is not relevant anymore. We have
        \begin{equation}
            \pseudothermtime{m} \sim N[m g(N)\pm O(1)], \label{eq:tptherm_pt3}
        \end{equation}
    for some function $\tilde{\Omega}(1) \leq g(N) \leq O(N)$. Depending on the scaling of $g(N)$, consecutive indistinguishability times $\pseudothermtime{m}$, $\pseudothermtime{m+1}$ may be sharply separated, or the two crossover windows may overlap. These two options are sketched in Fig.~\ref{fig:local}(c).
\end{enumerate}

\tocless\subsubsection{Higher dimension}

The analysis above could be straightforwardly adapted to higher-dimensional qubit arrays. For example one can choose a gate set $\mathcal{G} = \{u_{i,\hat{v},\hat{w},a,b}\}$, with $i\in[N]$ a site, $\hat{v},\hat{w}$ two different directions on the lattice, and $a,b \in \{0,1\}$ control bits, where the gate $u$ flips qubit $i$ iff its neighbors $i+\hat{v}$, $i+\hat{w}$ are in states $a,b$ respectively. $\mathcal{G}$ contaits $4\binom{2d}{2} N$ gates, with $d$ the spatial dimension. In $d=1$ one recovers the gate set used earlier (in that case the choice of directions $\hat{v}=+\hat{x}$, $\hat{w}=-\hat{x}$ is unique). 

For constant $d$, all the results of this section proceed unchanged, except for the locality bound, Eq.~\eqref{eq:tmix_locality_lowerbound}; there, we must replace Eq.~\eqref{eq:tmix_locality_lowerbound} with $t_{\rm mix}^{(m)} \geq \Omega(N^{1+1/d})$, with $d$ the spatial dimension. This still meets the requirement for a cut-off, Eq.~\eqref{eq:local_divergentratio}, but is not enough to pin down the scaling of indistinguishability time at constant $m$: we have instead $\Omega(N^{1+1/d}) \leq t_{\rm mix}^{(m)} \leq O(N^2)$. 

Qualitatively, this opens the question of whether, in dimension $d\geq 2$, the `pseudoentanglement front' picture still holds; i.e., if the state behind the light cone equilibrates quickly (which would give $t_{\rm mix}^{(m)} \sim N^{1+1/d}$), or whether it equilibrates more slowly, in such a way that even after the light cone has swept through the whole system it is still possible to distinguish the states from Haar-random ones with a constant number of copies $m$. 
We leave this question for future work.


\section{Generation of Pseudoentanglement}
\label{sec:generation}

\begin{figure}
    \centering
    \includegraphics[width=0.99\columnwidth]{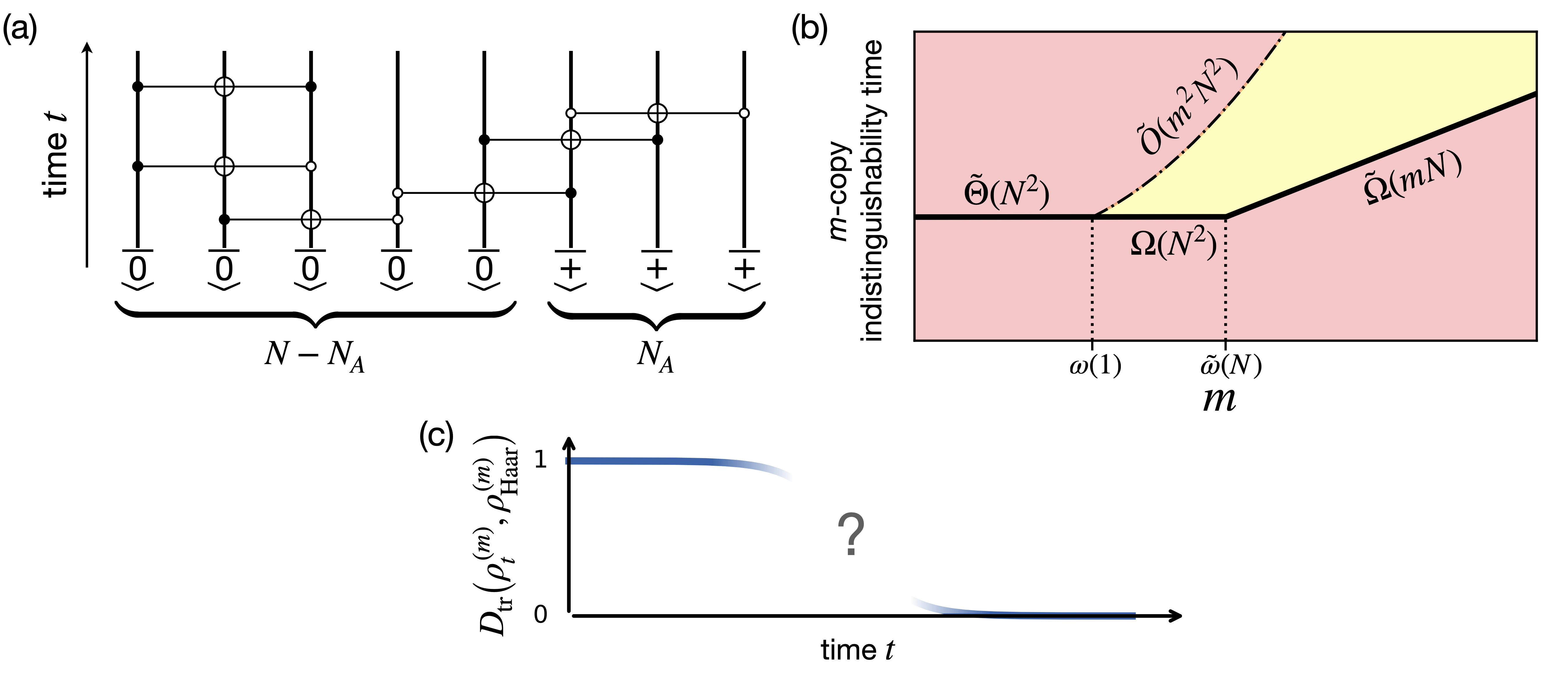}
    \caption{Local automaton circuit model for the generation of pseudoentanglement from product states.
    (a) Circuit model: the same as Fig.~\ref{fig:local}(a) but with a product initial state $\ket{0}^{\otimes N-N_A}\otimes \ket{+}^{\otimes N_A}$ (no random phases). 
    (b) Schematic summary of bounds for the $m$-copy indistinguishability time. Solid lines indicate unconditional bounds while the dot-dashed line depends on Conjecture~\ref{conjecture:relaxation}. Results are the same as in Fig.~\ref{fig:local}(b) but with a modified upper bound $\tilde{O}(m^2 N^2)$, Eq.~\eqref{eq:tptherm_upperbound}. 
    (c) Given the more complex mapping between quantum ensemble distinguishability and the mixing of classical Markov chains, Eq.~(\ref{eq:subsetstate_mapping1}-\ref{eq:subsetstate_mapping2}), the precise nature of ensemble pseudothermalization (e.g., whether there is a cut-off) remains unclear in this case. 
    }
    \label{fig:generation}
\end{figure}

Thus far we have discussed the propagation of pseudoentanglement: we have taken initial states that already contain random subset-phase ensembles on subsystems, and studied how these turn into random subset-phase states on the whole system under automaton circuit dynamics. This approach was motivated partly by analogy with the problem of entanglement spreading in quantum quenches, and partly by technical simplifications that arise when random phases are already provided by the input states.
In this section, we address the question of generating pseudoentanglement from initial product states. We use only the local gate set of Sec.~\ref{sec:local-circuit}.
To avoid technical complications associated with random phases, we choose to focus on the {\it random subset states} introduced in Refs.~\cite{giurgica-tiron_pseudorandomness_2023,jeronimo_pseudorandom_2024} and reviewed in Sec.~\ref{sec:review_pe}. 
We remark that this is a choice based on technical convenience. In principle an automaton circuit interspersed by phase gates will be able to produce random subset-phase states when acting on a trivial subset input states; however, we lack a statement like Propositions~\ref{proposition:subsetphase} (uniformly random phase functions) and \ref{proposition:subset} (constant phase functions) for the generic distribution of phase functions that would arise at intermediate times during that type of dynamics. However we note that the dynamics of subsets $S$ is independent of the phase functions $f$ in these automaton models, so convergence to random subset-phase states will be no faster than convergence to random subset states in the absence of phases. Thus it is reasonable to focus on random subset states in the following, while leaving the exploration of subset-phase state dynamics to future work.

Let us take an initial subset state 
\begin{equation}
    \ket{\psi_{S_0}} = \ket{0}^{\otimes N-N_A} \otimes \ket{+}^{\otimes N_A},
\end{equation}
corresponding to subset $S_0 = \{\mathbf z:\ z_1 = \dots = z_{N-N_A} = 0\}$ of cardinality $K = 2^{N_A}$, shown by Fig.~\ref{fig:generation}(a). 
To be pseudoentangled according to our definition (Sec.~\ref{sec:review_pe}), random subset ensembles need to have $\omega(\poly{N}) \leq K \leq 2^{o(N)}$, therefore we take the number of qubits initially in superposition, $k$, to be $\omega(\log N) \leq N_A < o(N)$.
After $t$ time steps, we have a subset state $\ket{\psi_{S_t}}$ where $S_t$ is a random variable on $\Sigma_K$ distributed according to a time-evolved distribution $p_t$.
To characterize the generation of pseudoentanglement, we look at the $m$-copy distinguishability from the Haar ensemble, i.e. the trace distance $D_{\rm tr}(\rho_{t}^{(m)}, \rho_{{\rm Haar}}^{(m)})$, with $\rho_t^{(m)} = \sum_S p_t(S)\ket{\psi_S}\bra{\psi_S}^{\otimes m}$.
As long as $m \leq O(\poly{N})$, the application of Proposition~\ref{proposition:subset} yields
\begin{align}
    D_{\rm tr}\left( \rho_{t}^{(m)}, \rho_{{\rm Haar}}^{(m)} \right)
    & \simeq  \frac{1}{2}  \left\| \mathcal{M}^{(m)}_{p_t} \right\|_{\rm tr},  \label{eq:subsetstate_mapping1} \\
    (\mathcal{M}^{(m)}_{p_t})_{S',S''} & = \binom{K}{m}^{-1} \binom{K}{m+\delta} f^{(m+\delta)}_{p_t} (S'\cup S''),    \label{eq:subsetstate_mapping2}
\end{align}
where $\delta \equiv |S'\setminus S''|$ is the difference between subsets $S',S''\in\Sigma_m$ (so $S'\cup S''\in \Sigma_{m+\delta}$), and 
\begin{equation}
    f^{(m+\delta)}_{p_t} = \Phi_{m+\delta \leftarrow K}[p_t] - \pi_{m+\delta}
    \label{eq:f_mdelta_def}
\end{equation}
is the deviation from uniformity of a Markov chain on $\Sigma_{m+\delta}$ induced by the random automaton circuit, as in Sec.~\ref{sec:spreading_pe}. 

In this case, the ensemble pseudothermalization problem is not directly mapped onto the mixing of a single classical random walk.
Rather, it is mapped 
in a nontrivial way, via Eq.~\eqref{eq:subsetstate_mapping2}, to the mixing of $m+1$ distinct random walks, indexed by $\delta = 0, 1, \dots m$, and distributed according to $\Phi_{m+\delta\leftarrow K}[p_t]$, defined by Eq.~\eqref{eq:phi_channel_def}. 

More specifically, the diagonal terms in $\mathcal{M}^{(m)}_{p_t}$ correspond to $\delta = 0$, i.e., to the walk on $\Sigma_m$ distributed as $\Phi_{m\leftarrow K}[p_t]$. This is the same walk that controlled the spreading of pseudoentanglement in Sec.~\ref{sec:spreading_pe}. 
However, the off-diagonal terms, corresponding to $\delta > 0$, make this problem more complex. 
While we do not have an exact solution in this case, one can still obtain nontrivial bounds on the indistinguishability time $\pseudothermtime{m}$, as we discuss next.

First, a lower bound on $\pseudothermtime{m}$ is obtained by using the monotonicity of trace norm under quantum channels: $\|\mathcal{D}(\mathcal{M})\|_{\rm tr}\leq \|\mathcal{M}\|_{\rm tr}$, where $\mathcal D$ is a dephasing super-operator that annihilates off-diagonal matrix elements: $[\mathcal D (\mathcal{M})]_{S',S''} = \delta_{S',S''} \mathcal{M}_{S',S'}$. This gives
\begin{equation}
    D_{\rm tr} \left( \rho_{t}^{(m)}, \rho_{{\rm Haar}}^{(m)} \right)
    \geq D_{\rm tv} \left( \Phi_{m\leftarrow K}[p_t], \pi_m \right) \label{eq:distance_lowerbound_diagonal}
\end{equation}
(up to negligible terms $O(m/\sqrt{K})$), which in turn implies
\begin{equation}
    \pseudothermtime{m} \geq t_{\rm mix}^{(m)}.
\end{equation}
Here $t_{\rm mix}^{(m)}$ is the mixing time of the Markov chain $\Phi_{m\leftarrow K}[p_t]$, studied in Sec.~\ref{sec:spreading_pe}; the lower bounds we derived there thus apply directly to the $m$-copy indistinguishability time in the present problem (in one spatial dimension):
\begin{equation}
    \pseudothermtime{m} \geq \Omega(N^2, mN/\log(N)).
    \label{eq:tptherm_lowerbound}
\end{equation}
Furthermore, since the matrix $\mathcal{M}^{(m)}_{p_t}$ that controls distinguishability, Eq.~\eqref{eq:subsetstate_mapping2}, consists of differences $\Phi_{m+\delta\leftarrow K}[p_t] - \pi_{m+\delta}$, we can guarantee ensemble pseudothermalization once all the walks ($\delta = 0,1,\dots m$) have come sufficiently close to equilibrium. 
By properly taking into account the multiplicity of each $\delta$, 
we derive the following upper bound on the $m$-copy indistinguishability time (see Appendix~\ref{app:subset_ptherm_bound}):
\begin{equation}
    \pseudothermtime{m} \leq O(\log(K) m^2 N^2) \leq \tilde{O}(m^2N^2),
    \label{eq:tptherm_upperbound}
\end{equation}
where the latter inequality holds if we take $K$ to be quasipolynomial, i.e., $\log K = {\rm polylog}(N)$.
Again, this is conditional on Conjecture~\ref{conjecture:relaxation}. 

Based on these bounds, summarized in Fig.~\ref{fig:generation}(b), we can conclude that
\begin{enumerate}[label = (\roman*)]
    \item at early times $t \leq \tilde{O}({\rm max}(N^2, mN))$, the system {\it can} be efficiently distinguished from Haar-random at the $m$-copy level. 
    \item at late times $t\geq \tilde{\Omega}(m^2N^2)$, the system {\it cannot} be efficiently distinguished from Haar-random at the $m$-copy level. 
\end{enumerate}
The nature of the crossover between these two regimes is an interesting open question, see Fig.~\ref{fig:generation}(c).
Under Conjecture~\ref{conjecture:relaxation}, we have $t^{(m)}_{\rm mix}/t^{(m)}_{\rm rel} \to \infty$ in this limit for all $m$, suggestive of a cutoff phenomenon in the associated Markov chains $\Phi_{m\leftarrow K}[p_t]$.  
How would such cutoff phenomena be reflected in the quantum distinguishability measure Eq.~\eqref{eq:subsetstate_mapping1}? Would the latter undergo an abrupt transition, would it decay gradually, or would it exhibit more complex behavior (e.g. a sequence of pre-cutoffs associated to each $\delta$)?
We leave this interesting question as a goal for future work.


\section{Discussion} 
\label{sec:discussion}

\subsection{Summary}

In this work we have examined the role of entanglement in the thermalization of isolated quantum many-body systems from a new perspective. 
We have introduced models of quantum dynamics that at late times reproduce all efficiently observable predictions of infinite-temperature quantum equilibrium, while only generating a relatively small amount of entanglement. 
This is contrary to the expectation that high-temperature equilibrium should be described by nearly-maximally-entangled random states, following the general ``maximum entropy principle'' that underpins thermodynamics. 
This counterintuitive result is made possible by the existence of pseudoentanglement, i.e., the ability of certain quantum state ensembles to securely hide their true entanglement structure from observers with limited resources. 

We have constructed random circuit models that provably equilibrate to ensembles of pseudoentangled states at late times, a phenomenon we have named ``ensemble pseudothermalization''. To understand this process quantitatively, we have studied the dynamical generation and propagation of pseudoentanglement in these circuit models through a mix of exact results, rigorous bounds, and numerical simulations. 
More specifically, our models are based on random reversible classical circuits that preserve certain types of states (subset-phase states~\cite{aaronson_quantum_2023} and subset states~\cite{giurgica-tiron_pseudorandomness_2023,jeronimo_pseudorandom_2024}) known to form pseudoentangled ensembles. 
We separately analyzed two related problems: the spreading of pseudoentanglement when a small pseudoentangled subsystem is placed in contact with a large trivial system, and the generation of pseudoentanglement starting from a product state. 

The key technical step in our analysis involves mapping the distinguishability between quantum state ensembles onto the equilibration of certain classical random walks over subsets of the computational basis. 
This creates a connection to an interesting feature of Markov chains on finite sets: the cutoff phenomenon, whereby mixing happens in a step-like fashion in the limit of large system size, giving asymptotically sharp transitions to equilibrium. We argued that our Markov chains in many cases exhibit a cutoff phenomenon, giving sharp dynamical transitions in the distinguishability of our output states from Haar-random ones.

The new concept of ensemble pseudothermalization adds to a growing body of research on the nuances of thermalization, chaos, ergodicity, scrambling, {\it etc.}, that can be revealed only by higher-moment statistical quantities, thus going beyond the first- or second-moment quantities (e.g., local expectation values~\cite{srednicki_chaos_1994}, out-of-time-ordered correlators~\cite{von_keyserlingk_operator_2018}) that are traditionally the basis of such notions. 
Other recently introduced ideas in this direction include {\it deep thermalization}~\cite{cotler_emergent_2023,ho_exact_2022,ippoliti_dynamical_2023,ippoliti_solvable_2022,lucas_generalized_2023,bhore_deep_2023,shrotriya_nonlocality_2023,chan_projected_2024}, {\it complete Hilbert-space ergodicity}~\cite{pilatowsky-cameo_complete_2023,pilatowsky-cameo_hilbert-space_2024}, and thermalization via {\it free probability} or {\it cumulants}~\cite{fava_designs_2023,pappalardi_general_2023}. 
Those notions tend to be stronger than standard thermalization (as defined from the equilibration of local reduced density matrices). They are based on information-theoretic pseudorandomness of certain state ensembles generated for example by quantum measurements or by quasiperiodic drive sequences. The randomness can be quantified through the notion $\epsilon$-approximate quantum state $m$-design, corresponding to the trace distance between $m$-th moment operators being smaller than a tolerance $\epsilon$ (as in Eq.~\eqref{eq:pe_criterion_tracedistance}); in those stronger versions of thermalization, one generally requires or expects an exponentially small tolerance $\epsilon \sim 1/\exp(N)$, or possibly even exact state designs ($\epsilon = 0$). For $m\geq 2$, these conditions imply high entanglement.

Ensemble pseudothermalization in contrast is based on a computational notion of pseudorandomness, where the observer is limited to polynomial resources (e.g. time, number of state copies) in their attempt to tell the states apart from truly random ones. Again in the language of state designs, this corresponds to a tolerance $\epsilon = o(1/\poly{N})$ that is superpolynomially small (in order to trick the polynomially-bound observer), but not exponentially small, thus allowing for weakly entangled states. Moreover, while we have used the trace distance as a measure of distinguishability in this work, in reality the observer might be limited to sub-optimal measurements that do not saturate the trace distance bound (the optimal distinguishing measurement might require a deep quantum circuit to implement). It would be interesting to explore how our results would change when keeping this additional limitation into account.

\subsection{Implications} 

Our results carry several implications for quantum physics and quantum information science. We summarize them below, while noting that fully unraveling their consequences will require future work.

{(i) \it Learning quantum processes and entanglement growth.}
Ensemble pseudothermalization implies limits on the ability of an observer to efficiently characterize entanglement growth in many-body dynamics, an important tool in experimental studies of thermalization~\cite{kaufman_quantum_2016}. While pseudoentanglement makes this statement for individual states, our work translates this statement to time evolution: it is possible to efficiently verify, e.g., the linear growth of quantum entanglement versus time $t$ (a feature of quantum chaotic dynamics) up to a point, $t\sim\log N$, but not beyond it. 
Our results also have broader implications on the learning of quantum processes. While our automaton circuits are {\it not} pseudorandom (evident from their trivial action on computational basis input states), they would nonetheless appear random to several protocols for quantum process learning~\cite{huang_learning_2023,levy_classical_2024}. These protocols are based on feeding in random Pauli product states (randomly $\pm 1$ in the $X$, $Y$ or $Z$ basis on each qubit) and then performing randomized Pauli-basis measurements on the output. Since random Pauli product states are subset-phase states of typical cardinality $K \sim 2^{(2/3)N}$ (each $X$ or $Y$ basis state doubles the cardinality of the subset), this ensemble of input states can pseudothermalize for any polynomial number of copies $m$. Classical shadows on the output states will not be able to tell the state apart from Haar states, and thus the process cannot be efficiently distinguished from a Haar-random unitary by these methods.

{\it (ii) Detection of measurement-induced phenomena.} 
Much research has been devoted lately to phenomena that take place in conditional, post-measurement states of quantum many-body systems. These include measurement-induced entanglement and teleportation transitions~\cite{bao_finite-time_2024,hoke_measurement-induced_2023} as well as `deep thermalization', the emergence of universal wavefunction distributions on subsystems from measurements on the rest of the system (the `projected ensemble')~\cite{cotler_emergent_2023,ho_exact_2022}. 
Given the intrinsically non-deterministic nature of measurement, detection of all these phenomena na\"ively relies on postselection, which is inefficient for many-body states; however, efficient detection may become possible when allowing adaptive and computationally-enhanced strategies~\cite{gullans_scalable_2020,noel_measurement-induced_2022,hoke_measurement-induced_2023,garratt_probing_2024,mcginley_postselection-free_2024}. 
Ensemble pseudothermalization presents universal limits on these strategies.
The subset-phase state ensembles obtained as fixed points in our dynamics lead to very simple post-measurement states: when a finite fraction of the system is measured in the computational basis, one  almost always obtains a product state on the remainder of the system\footnote{
Given a subset state $\ket{\psi_S}$, measuring a bitstring $\zeta$ on subsystem $\bar{A}$ collapses the state to a superposition of all computational basis states $\ket{\eta}_A$ on $A$ such that the combined bitstring $z=(\eta,\zeta)$ belongs to the original subset $S$; but this typically specifies a unique $\eta$. The probability of two distinct $\eta,\eta'$ appearing in the post-measurement state is the probabilty that two bitstrings $z,z'\in S$ agree on the entire subsystem $\bar{A}$, which is $O(K^2 2^{-|\bar A|})$ (by birthday asymptotics) for a random subset state. As long as $K = 2^{o(N)}$ and $|\bar{A}| = \Theta(N)$ this probability is exponentially small in $N$.
}.
This means, e.g., that our pseudoentangled states are in the trivial phase for measurement-induced teleportation, and have a trivial projected ensemble. One the contrary, Haar-random states are in the non-trivial (teleporting) phase~\cite{bao_finite-time_2024} and their projected ensemble can form higher state designs~\cite{cotler_emergent_2023}.
Thus any method to efficiently diagnose these facts (without prior knowledge of the state) would serve as an efficient protocol to distinguish pseudoentangled states from Haar-random states, which is by definition impossible. It follows that any efficient protocol for the detection of these measurement-induced phenomena {\it must} use prior information about the state preparation~\cite{mcginley_postselection-free_2024}. Without this information, no efficient protocol (even allowing for arbitrary adaptivity and multi-copy measurements) can succeed.

{\it (iii) Fundamental physics implications.} 
Our results highlight the fact that extensive amounts of entanglement are generally not necessary to reproduce the predictions of high-temperature thermal equilibrium. The results on the generation of pseudoentanglement from product states are especially remarkable: they show that a disentangled product state, featuring only a modest amount (vanishing density) of coherence and acted upon by `classical' dynamics and [see Fig.~\ref{fig:generation}(a)], becomes indistinguishable from a near-maximally entangled random state. This raises the prospect that the universe might, in reality, be ``less quantum'' than we think---for example governed by mostly classical local update rules, enhanced by relatively small amounts of quantum resources like entanglement and coherence. Theory proposals of ``emergent quantum mechanics''~\cite{slagle_testing_2022} may include a cut-off on the allowed amount of entanglement in many-body systems, or present deviations from standard quantum mechanics only in highly entangled states. Our findings suggest that it may be impossible to efficiently rule out such scenarios in general. 

\subsection{Outlook}

Our work opens a number of directions for future research. 
First of all, it would be desirable to improve and/or rigorize our results in several ways: by tightening bounds, especially on the mixing times in the local circuit model and on pseudoentanglement without random phases; 
by proving or disproving our Conjecture~\ref{conjecture:relaxation} on the relaxation times in the local circuit model; 
and by proving the existence of cutoff phenomena in our Markov chains, for which we only showed a necessary condition (divergent ratio of mixing time to relaxation time).
These improvements could have significant qualitative consequences. 
For example, one may hope to show that, in the local circuit model for pseudoentanglement spreading, Sec.~\ref{sec:local-circuit}, the $m$-copy distinguishability transitions are sharply separated from one another, in the sense that the difference of consecutive mixing times $t_{\rm mix}^{(m+1)} - t_{\rm mix}^{(m)}$ is sufficiently larger than the relaxation time $t_{\rm rel}^{(m)}$. This would give a cascade of distinguishability transitions, one for each number of copies $m$ that the observer can access, sharply separated in time from each other.

Another direction for future work is to identify physically realistic models of dynamics where pseudoentanglement is achieved in polynomial time. 
In the all-to-all model of Sec.~\ref{sec:non-local_gate} that time scale is exponential for trivial reasons;
in the local model of Sec.~\ref{sec:local-circuit}, the $m$-copy indistinguishability time is polynomial in system size, but grows at least linearly with $m$, so that convergence to a pseudoentangled ensemble, indistinguishable from Haar-random with {\it any} polynomial number of copies $m = O(\poly{N})$, requires super-polynomial time. This is a consequence of demanding information-theoretic pseudorandomness, i.e. vanishing trace distance [Eq.~\eqref{eq:pe_criterion_tracedistance}], along with a gate set of polynomial size (this follows from the counting bound, Fact~\ref{fact:countingbound}). 
It would be interesting to identify physically reasonable models of dynamics that generate pseudoentanglement in polynomial time, under the weaker condition of computational pseudorandomness [Eq.~\eqref{eq:pe_criterion_algo}]. 
Indeed pseudorandom subset(-phase) state ensembles can be prepared efficiently by using quantum-secure pseudorandom permutations and phase functions~\cite{aaronson_quantum_2023}. 
In this work we have considered truly random circuits, as these make for simpler models of dynamics and are more in keeping with the spirit of quantum thermalization; however it would be interesting to identify pseudorandom models of dynamics that can produce pseudoentanglement in polynomial time while maintaining a simple, local, ``random-looking'' structure. A related construction of pseudorandom unitaries was recently shown to not be robust to perturbations~\cite{haug_pseudorandom_2023}, which might present an obstruction for this direction.

More broadly, subset(-phase) states are only one class of pseudoentangled state ensembles. Due to their structure, these ensembles suggest automaton circuits as a natural models of dynamics. It would be interesting, given different ensembles of pseudoentangled states (for instance a recent construction based on shallow circuits of pseudorandom unitaries acting on large blocks of qubits~\cite{schuster_random_2024}), to identify models of quantum dynamics that might naturally produce them as late-time steady states. 

Recent works have introduced related notions of `pseudocoherence'~\cite{haug_pseudorandom_2023}, `pseudomagic'~\cite{gu_little_2023}, etc, which are analogous to pseudoentanglement but based on different resources (coherence, magic, etc). A natural extensions of our work would be to study the dynamics of other such pseudoresources in appropriate circuit models. 

It is also interesting to ask whether, in analogy with measurement-induced entanglement phase transitions~\cite{li_quantum_2018,skinner_measurement-induced_2019,gullans_dynamical_2020,iaconis_measurement-induced_2020}, one might obtain phase transitions in the dynamics of pseudoentanglement by enriching our automaton circuits with suitable projective measurements and other elements. For example, while the size of subsets is invariant under automaton gates, it generally decreases under $Z$ measurements, and can increase under Hadamard gates and $X$ measurements. Can the competition of these elements give rise to nontrivial phase transitions in the pseudoentanglement properties of late-time states? Such transitions might carry implications for the hardness of detecting general measurement-induced phenomena, as discussed earlier. 

Finally, the possible emergence of cutoff phenomena in ensemble pseudothermalization raises more general questions about the sharpness of thermalization in generic quantum dynamics. Are sharp transitions in distinguishability a generic feature of chaotic quantum dynamics~\cite{kastoryano_cutoff_2012,oh_cutoff_2023,balasubramanian_glassy_2024}? What are the right tools to analyze them? Can particularly structured classes of time evolutions (e.g., with symmetry) exhibit different levels of sharpness? We leave these questions as an interesting target for future research.

{\it Note added.} Shortly after the first version of this manuscript appeared on arXiv, two new results on the spectral gaps of random reversible classical circuits have appeared. Ref.~\cite{he_pseudorandom_2024} studied a model of geometrically local 3-bit gates very similar to ours and proved (in our notation) an upper bound on the relaxation time $t_{\rm rel} \leq \tilde{O}(mN)$. 
Subsequently, Ref.~\cite{chen_incompressibility_2024} considered 3-local, but geometrically non-local, random reversible circuits and showed an upper bound $t_{\rm rel} \leq \tilde{O}(N)$ for all $m = O(\poly{N})$. 
They conjecture that the bounds should apply to the geometrically-local setting as well, which would prove our Conjecture~\ref{conjecture:relaxation}.

\begin{acknowledgments} 
We thank Scott Aaronson, Tudor Giurgica-Tiron, Nick Hunter-Jones, Ethan Lake, and Shivan Mittal for stimulating discussions. 
We are especially grateful to Nick Hunter-Jones for pointing us to Ref.~\cite{gowers_almost_1996} and to Shivan Mittal for sharing his results~\cite{mittal_private_2024} on the gap of a Hamiltonian discussed in Appendix~\ref{app:hypercube_rw_relative}.
Some numerical simulations were carried out on HPC resources provided by the Texas Advanced Computing Center (TACC) at the University of Texas at Austin. 
\end{acknowledgments}

\bibliographystyle{JHEP}
\bibliography{pseudoentanglement.bib}

\appendix

\newpage



\section{Proof of Proposition~\ref{proposition:subsetphase}}
\label{app:proof1}

Here we prove Proposition~\ref{proposition:subsetphase} by analyzing the trace distance
\begin{equation}
    D_{\rm tr}\left( \rho^{(m)}_{\mathcal E_p}, \rho^{(m)}_{\rm Haar} \right) = \frac{1}{2} \left\| \rho^{(m)}_{\mathcal E_p} - \rho^{(m)}_{\rm Haar} \right\|_{\rm tr} 
    \label{eq:app_tracedistance}
\end{equation}
between the $m$-th moment operators for the subset-phase ensemble $\mathcal{E}_p$ (with uniformly random phases but arbitrarily-distributed subsets, $S\sim p$) and the Haar ensemble.

\tocless\subsection{Preliminaries}

In this Appendix we will denote computational basis states by integers $z\in [D] = \{0,1,\dots D-1\}$ rather than by bitstrings $\mathbf{z}\in \{0,1\}^N$. We will reserve the latter notation for tuples of computational basis states, which are used later.

Eq.~\eqref{eq:app_tracedistance} involves states in the symmetric sector of the $m$-fold replicated Hilbert space $\mathcal{H}^{\otimes m}$. It is thus helpful to introduce an orthonormal basis of that subspace:
\begin{equation}
    \ket{\mathbf T} \equiv \binom{m}{\mathbf T}^{-1/2} \sum_{\mathbf z:\, {\rm type}(\mathbf z) = \mathbf T}\ket{\mathbf z},
    \label{eq:type_basis}
\end{equation}
where $\mathbf z \in [D]^m$ is an $m$-tuple of basis states and ${\rm type}(\mathbf z) \in [m+1]^D$ is a vector that counts the multiplicity of each basis state $z$ in the $m$-tuple $\mathbf z$ (the notation is borrowed from Ref.~\cite{aaronson_quantum_2023}): i.e., if ${\rm type}(\mathbf z) = \mathbf T$, then $T_z$ is the number of times the basis state $z \in [D]$ appears in $\mathbf{z}$. As a consequence we have $\sum_{z=0}^{D-1} T_z = m$. Finally, we used a shorthand for the multinomial coefficient 
$$\binom{m}{\mathbf T} = \binom{m}{T_0,\dots T_{D-1}} = \frac{m!}{T_0! \cdots T_{D-1}!}.$$
Types that consist only of 0's and 1's are called ``unique types''. These describe non-degenerate $m$-tuples $\mathbf z$, where all elements are distinct. As such, unique types are equivalent to subsets of cardinality $m$: $\mathbf{T}\leftrightarrow S = \{z:\, T_z=1\}$. With slight abuse of notation we will denote a unique-type state $\ket{\mathbf T}$ by the corresponding subset, $\ket{S}$.

We may write the Haar ensemble's moment in terms of type states as
\begin{equation}
    \rho^{(m)}_{\rm Haar} = \binom{D+m-1}{m}^{-1} \sum_{\mathbf T\in {\rm types}} \ketbra{\mathbf T}.
\end{equation}
Further, it is useful to introduce the ``unique-type state''
\begin{align}
    \rho_{{\rm unique}}^{(m)} 
    & = \binom{D}{m}^{-1} \sum_{\mathbf{T} \in \text{unique types}} \ketbra{\mathbf T} \nonumber \\
    & = \binom{D}{m}^{-1} \sum_{S'\in\Sigma_m} \ketbra{S'}.
\end{align}
It is easy to verify that, as $D\to\infty$ with finite $m$, most types become unique, and $\rho_{\rm Haar}^{(m)} \simeq \rho_{\rm unique}^{(m)}$. 
We also introduce the same object restricted to basis states within a particular subset $S\in\Sigma_K$:
\begin{align}
    \rho_{S, {\rm unique}}^{(m)} 
    & = \binom{K}{m}^{-1} \sum_{S'\in\Sigma_m:\, S'\subset S} \ketbra{S'}.
    \label{eq:uniquetypeprojS}
\end{align}

\tocless\subsection{Proof}

Our proof builds on Ref.~\cite{aaronson_quantum_2023}, specifically the two facts below:

\begin{fact}
We have
$$ \left\| \rho_{\rm Haar}^{(m)} - \sum_{S \in \Sigma_K} \pi_K(S) \rho_{S,{\rm unique}}^{(m)} \right\|_{\rm tr} \leq O\left( \frac{m^2}{K} \right), $$
with $\pi_K(S) = \binom{D}{K}^{-1}$ the uniform distribution on $\Sigma_K$. 
\end{fact}

\begin{proof} 
See Proposition 2.3 in Ref.~\cite{aaronson_quantum_2023}.
\end{proof}

\begin{fact}
We have
$$ \left\| \rho_{\mathcal{E}_p}^{(m)} - \sum_{S \in \Sigma_K} p(S) \rho_{S,{\rm unique}}^{(m)} \right\|_{\rm tr} \leq O\left( \frac{m^2}{K} \right), $$
\end{fact}

\begin{proof} 
See Proposition 2.4 in Ref.~\cite{aaronson_quantum_2023}.
\end{proof}

By using these two facts and the triangle inequality, we can rephrase the LHS of Eq.~\eqref{eq:app_tracedistance} as 
\begin{align}
    \left\| \rho^{(m)}_{\mathcal E_p} - \rho^{(m)}_{\rm Haar} \right\|_{\rm tr} 
    & = \left\|\sum_{S\in \Sigma_K} (p(S)-\pi_K(S))\rho_{S,{\rm unique}}^{(m)} \right\|_{\rm tr} 
    + O \left(\frac{m^2}{K}\right). 
\end{align}
Now, using the definition of unique-type state Eq.~\eqref{eq:uniquetypeprojS} and orthonormality of the $\{\ket{S'}\}_{S'\in \Sigma_m}$ states, we arrive at 
\begin{align}
    \left\| \rho^{(m)}_{\mathcal E_p} - \rho^{(m)}_{\rm Haar} \right\|_{\rm tr} 
    & =  \sum_{S' \in \Sigma_m} | f(S')| + O \left(\frac{m^2}{K}\right), \\
    f(S')
    & = \binom{K}{m}^{-1} \sum_{\substack{S\in\Sigma_K: \\ S'\subset S}} [p(S) - \pi_K(S)]
\end{align}
Finally, we obtain the statement of Proposition~\ref{proposition:subsetphase} by introducing the map $\Phi_{K\leftarrow m}$ of Eq.~\eqref{eq:phi_channel_def}.


\section{Properties of the $\Phi$ maps}
\label{app:phi_maps}

Here we briefly analyze and prove some properties of the $\Phi$ maps.

\begin{fact}
    $\Phi_{m\leftarrow K}$ maps probability distributions on $\Sigma_K$ to probability distributions on $\Sigma_m$.
\end{fact}

\begin{proof}
It is straightforward to verify that $\Phi_{m\leftarrow K}[p]$ corresponds to the distribution on $\Sigma_m$ obtained by composing these two processes: (i) draw $S\in \Sigma_K$ according to $p$; (ii) draw (without replacement) $m$ elements from $S$ uniformly at random.
\end{proof}

\begin{fact}
\label{fact:semigroup}
    The maps $\Phi$ define a semigroup: 
    \begin{equation}
        \Phi_{c\leftarrow b} \circ \Phi_{b\leftarrow a} = \Phi_{c\leftarrow a}
    \end{equation}
    for all $D\geq a> b> c\geq 1$. 
\end{fact}

\begin{proof}
We have
\begin{align}
    \Phi_{c\leftarrow b} \circ \Phi_{b\leftarrow a}[p](S'')
    & = \binom{a}{b}^{-1} \binom{b}{c}^{-1} \sum_{\substack{S,S':\\ S''\subset S'\subset S}} p(S) \nonumber \\
    & = \binom{a}{b}^{-1} \binom{b}{c}^{-1} \binom{a-c}{b-c} \binom{a}{c} 
    \times \Phi_{c\leftarrow a}[p](S''),
\end{align}
where in the second line we carried out the summation over $S'$ (the $c$ elements of $S'\cap S''$ are fixed, whereas the $b-c$ elements of $S'\setminus S''$ must be chosen from $S\setminus S''$, which has $a-c$ elements). Straightforward algebra shows the product of binomial coefficients is 1. 
Alternatively, without algebra, one can note that, given a set $S$ with $|S|=a$, the following two processes are equivalent: (i) draw a uniformly-random subset $S'\subset S$ with $|S'|=b$, and then draw a uniformly-random subset $S''\subset S'$ with $|S''|=c$;
(ii) draw a uniformly-random subset $S''\subset S$ with $|S''|=c$. 
\end{proof}

Next we focus on the relationship between the distributions $\Phi_{m\leftarrow K}[p]$ with different values of $m$. Let us consider the relevant situation of a random walk on the permutation group $\perm{D}$, represented by a probability distribution $\mathbb{P}_t(\sigma)$. This induces random walks on the spaces $\Sigma_m$ for all $m$: picking an initial subset $S_0\in\Sigma_m$, we can define 
\begin{equation} 
p_{t,S_0}^{(m)}(S) = \sum_{\sigma\in\perm{D}:\, \sigma(S_0)=S} \mathbb{P}_t(\sigma), \label{eq:app_induced_walk}
\end{equation} 
with $\sigma(S_0)$ representing the action of permutations on subsets.

\begin{fact}
    Given a Markov chain $p_{t,S_0}^{(K)} (S)$ on $\Sigma_K$ induced by a random walk on permutations as in Eq.~\eqref{eq:app_induced_walk}, we have
    \begin{equation}
        \Phi_{m\leftarrow K}\left[p_{t,S_0}^{(K)}\right](S') = \mathbb{E}_{S_0'}\left[ p^{(m)}_{t,S_0'}(S')\right],
        \label{eq:derived_markovchain}
    \end{equation}
    where the average is taken over subsets $S_0'\in \Sigma_m$, $S_0'\subset S_0$, according to the uniform measure.
\end{fact}

\begin{proof} 
Plugging in the definition of $\Phi$, the LHS of Eq.~\eqref{eq:derived_markovchain} reads
\begin{align}    
\Phi_{m\leftarrow K}\left[p_{t,S_0}^{(K)}\right](S') 
& = \binom{K}{m}^{-1} \sum_{S:\, S'\subset S } p_{t,S_0}^{(K)}(S) \nonumber \\
& = \binom{K}{m}^{-1} \sum_{S:\, S'\subset S } \sum_{\substack{\sigma \in \perm{D}:\\ \sigma(S_0) = S}} \mathbb{P}_t(\sigma) \nonumber \\
& = \binom{K}{m}^{-1} \sum_{\substack{\sigma \in \perm{D}: \\ S'\subset \sigma(S_0)}} \mathbb{P}_t(\sigma).
\end{align}    
But the condition $S'\subset \sigma(S_0)$ is verified if and only if $S' = \sigma(S_0')$ for some subset $S_0'\subset S_0$. We therefore obtain 
\begin{align}
    \Phi_{m\leftarrow K}\left[p_{t,S_0}^{(K)}\right](S')  & = \binom{K}{m}^{-1} \sum_{S_0'\subset S_0} \sum_{\substack{\sigma \in \perm{D}:\\ S' = \sigma(S_0')}} \mathbb{P}_t(\sigma)  \nonumber \\
    & = \binom{K}{m}^{-1} \sum_{S_0'\subset S_0} p^{(m)}_{t,S_0'}(S')
\end{align}
where we have recognized the definition of $p_{t,S_0'}^{(m)}(S')$, the Markov chain induced on $\Sigma_m$, with initial state $S_0'$. We see that $S_0'$ is averaged uniformly over all subsets of $S_0$ of the appropriate cardinality, as in Eq.~\eqref{eq:derived_markovchain}. 
\end{proof}

\begin{fact}[monotonicity] \label{fact:monotonicity}
The total variation distances 
\begin{equation}
    \Delta_m \equiv D_{\rm tv}( \Phi_{m\leftarrow K}[p], \pi_m )
\end{equation}
are non-increasing in $m$. 
\end{fact}

\begin{proof}
We have
    \begin{align}
        \Delta_{m-1} 
        & = \frac{1}{2} \|\Phi_{m-1\leftarrow K}[p-\pi_K] \|_{\rm tv} \nonumber \\
        & = \frac{1}{2}  \|\Phi_{m-1\leftarrow m} \circ \Phi_{m\leftarrow K}[p-\pi_K] \|_{\rm tv} \nonumber \\
        & \leq \frac{1}{2}  \|\Phi_{m\leftarrow K}[p-\pi_K] \|_{\rm tv} = \Delta_m. 
    \end{align}
The first line uses the fact that $\Phi_{m\leftarrow K}[\pi_K] = \pi_{m}$ for all $m<K$; the second uses the semigroup law of $\Phi$ maps, Fact~\ref{fact:semigroup}; the last one uses monotonicity of total variation distance under stochastic maps: $D_{\rm tv}(\Phi(p),\Phi(q))  \leq D_{\rm tv}(p,q)$. 
\end{proof}


\section{Proof of Proposition~\ref{proposition:subset}}
\label{app:proof2}

We begin by recalling a useful fact proven in Ref.~\cite{giurgica-tiron_pseudorandomness_2023}:
\begin{fact} For $\omega(\poly{N}) < K < o(2^N)$, we have
    \begin{equation}
        \left\| \sum_{S\in \Sigma_K} \pi_K(S)\ketbra{\psi_S}^{\otimes m} - \rho_{\rm Haar}^{(m)} \right\|_{\rm tr} \leq O\left(\frac{mK}{D}, \frac{m^2}{K} \right).
    \end{equation}
\end{fact}

\begin{proof}
See Theorem~1 in Ref.~\cite{giurgica-tiron_pseudorandomness_2023}. 
\end{proof}

Using this fact, in the relevant range of $K$, we can recast the LHS of Proposition~\ref{proposition:subset} in the more convenient form
\begin{equation}
    \frac{1}{2} \left\|\sum_{S}[p(S)-\pi_K (S)]\ketbra{\psi_S}^{\otimes m}\right\|_{{\rm tr}}
    \label{eq:app_subset_target}
\end{equation}
up to small error $O(mK/D,m^2/K)$. 

Next, using the type-state basis introduced in Appendix~\ref{app:proof1}, Eq.~\eqref{eq:type_basis}, we can express the tensor power $\ket{\psi_S}^{\otimes m}$ as
\begin{equation}
    \ket{\psi_S}^{\otimes m}=\frac{1}{K^{m/2}}\sum_{\mathbf{T}\in \mathcal{T}^m_S}\binom{m}{\mathbf{T}}^{1/2}\ket{\mathbf{T}}.
\end{equation}
Here we use $\mathcal{T}_S^m$ to represent all types of $m$-tuples that have nonzero components only for states in $S$. 
We also define pure {unique-type} states by
\begin{align}
\ket{\psi_{S,{\rm uni}}^{(m)}} =\binom{K}{m}^{-1/2}\sum_{{S'\in \Sigma_m: \, S' \subset S}} \ket{S'},
\label{eq:uniquetypepurestates}
\end{align}
where $\ket{S'}$ stands for the type state $\ket{\mathbf{T}}$ where $\mathbf{T}$ is the unique type whose 1's are on the states in $S'$ ($T_z=1$ iff $z\in S'$, $T_z = 0$ otherwise).
As also noted in Appendix~\ref{app:proof1}, in the limit of $D\gg m$ most types are unique.
This is helpful as it allows us to replace the state $\ket{\psi_S}$ in Eq.~\eqref{eq:app_subset_target} by its unique-type counterpart, $\ket{\psi_{S,{\rm uni}}}$, with a small error. To show this, we first need a simple fact about the trace distance between pure states:

\begin{fact} \label{fact:tracedistance_purestates}
For any two pure states $\ket{\phi}$, $\ket{\chi}$ we have
    \begin{equation}
        \| \ketbra{\phi} - \ketbra{\chi}\|_{\rm tr} = 2\sqrt{1-|\braket{\phi}{\chi}|^2}.
    \end{equation}
\end{fact}

\begin{proof}
    See Fact 2 in Ref.~\cite{aaronson_quantum_2023} (set $\rho_1 = \ketbra{\phi}$, $\rho_2 = \ketbra{\chi}$).
\end{proof}

We can now show the following:
\begin{lemma}[unique types suffice] 
We have
    \begin{equation}
        \left\|\ketbra{\psi_S}^{\otimes m}-\ketbra{\psi_{S,{\rm uni}}^{(m)}}\right\|_{{\rm tr}}
        \leq O\left( \frac{m}{\sqrt{K}} \right). \label{eq:app_replace_S_unique}
    \end{equation}
\end{lemma}

\begin{proof}
    Using Fact~\ref{fact:tracedistance_purestates}, the LHS reads
    \begin{align}
        2\sqrt{1-\left|\braket{\psi_S^{\otimes m}}{\psi_{S,{\rm uni}}^{(m)}}\right|^2}\nonumber.
    \end{align}
    Writing both $\ket{\psi_S^{\otimes m}}$ and $\ket{\psi_{S,{\rm uni}}^{(m)}}$ in their type basis representation, we get the inner product
    \begin{align}
        \braket{\psi_S^{\otimes m}}{\psi_{S,{\rm uni}}^{(m)}}
        & = K^{-m/2} \binom{K}{m}^{-1/2} 
        \sum_{\mathbf T \in \mathcal{T}_S^m} \binom{m}{\mathbf T}^{1/2} \delta(\mathbf T \text{ is unique}) \nonumber \\
        & = (m!)^{1/2} K^{-m/2} \binom{K}{m}^{1/2}.
    \end{align}
    Using some straightforward algebra and the `birthday asymptotics' $\frac{K!}{(K-m)!} = K^m \left( 1 + O(m^2/K)\right)$ (see e.g. Fact 1 in Ref.~\cite{giurgica-tiron_pseudorandomness_2023}) we finally get
    \begin{align}
        \sqrt{1-\left|\braket{\psi_S^{\otimes m}}{\psi_{S,{\rm uni}}^{(m)}}\right|^2}
        & = O\left(\frac{m}{\sqrt K}\right).
    \end{align}
\end{proof}

Putting it all together, we can rewrite Eq.~\eqref{eq:app_subset_target} as
\begin{align}
    \frac{1}{2} \left\| \sum_{S \in \Sigma_K} (p(S)-\pi_K(S)) \ketbra{\psi_{S,{\rm uni}}^{(m)}} \right\|_{\rm tr} 
\end{align}
now up to error $O(m/\sqrt{K})$. Finally, plugging in the definition of unique-type pure states, Eq.~\eqref{eq:uniquetypepurestates}, yields (up to the $1/2$ prefactor)
\begin{align}
    & \left\| \sum_{S \in \Sigma_K} (p(S)-\pi_K(S)) \binom{K}{m}^{-1} \sum_{\substack{S',S''\in \Sigma_m:\\ S',S''\subset S}} \ketbra{S'}{S''} \right\|_{\rm tr} \nonumber \\
    = & \left\| \sum_{S',S''\in \Sigma_m} \ketbra{S'}{S''} \left. \binom{K}{m}^{-1} \sum_{\substack{S\in\Sigma_K:\\ S'\cup S''\subset S}} (p(S)-\pi_K(S)) \right. \right\|_{\rm tr}. 
\end{align}
We can recognize the coefficient of $\ketbra{S'}{S''}$ as 
\begin{equation}
    \binom{K}{m}^{-1} \binom{K}{m+\delta} \Phi_{m+\delta\leftarrow K}[p-\pi_K](S'\cup S'')
\end{equation}
which is the matrix $\mathcal{M}_p^{(m)}$ in Proposition~\ref{proposition:subset}.


\section{Random walk on the hypercube} 

\subsection{Single particle} \label{app:hypercube_rw}

In this Appendix we review, for completeness, the solution of the simple random walk on the hypercube $\mathbb{Z}_2^N$~\cite{diaconis_asymptotic_1990}. We label each vertex of the hypercube by a bitstring $\mathbf{z} \in \{0,1\}^N$; two vertices $\mathbf x,\mathbf y$ are neighbors iff the corresponding bitstrings have Hamming distance 1: $| \mathbf{x}\oplus \mathbf{y}'| = 1$ ($\oplus$ is sum modulo 2). 
In each time step the walker hops randomly in any one of the $N$ directions, represented by the Markov chain transition matrix
\begin{equation}
    \Gamma_{\mathbf x, \mathbf y} = \frac{1}{N} \delta_{|\mathbf x \oplus \mathbf y|, 1}.
    \label{eq:srw_hypercube_gamma}
\end{equation}
Note that as written, the walk is not aperiodic: the parity of the Hamming weight alternates in each time step (this is a general feature of bipartite graphs). This problem can be eliminated by allowing idling, i.e. replacing $\Gamma_{\mathbf x,y} \mapsto p \Gamma_{\mathbf x,y} + (1-p) \delta_{\mathbf x, \mathbf y}$ (the walker takes a step with probability $p$, remains in place otherwise). 
This does not change the eigenfunctions and changes the eigenvalues in a trivial way, $\lambda \mapsto p \lambda + (1-p)$. 

The transition matrix in Eq.~\eqref{eq:srw_hypercube_gamma} is diagonalized via the Hadamard transform, $H_{\mathbf a, \mathbf x} \equiv \frac{1}{\sqrt{D}} (-1)^{\mathbf a \cdot \mathbf x}$:
\begin{align} 
\tilde{\Gamma}_{\mathbf a,\mathbf b} 
& \equiv \frac{1}{D} \sum_{\mathbf x, \mathbf y} \Gamma_{\mathbf x, \mathbf{y}} (-1)^{\mathbf a\cdot \mathbf x + \mathbf b\cdot \mathbf y} \nonumber \\
& = \frac{1}{ND} \sum_{\mathbf x, \mathbf r} \delta_{|\mathbf r|-1} (-1)^{(\mathbf a\oplus \mathbf b)\cdot \mathbf x + \mathbf b\cdot \mathbf r} \nonumber \\
& = \delta_{\mathbf a, \mathbf b} \frac{1}{N} \sum_{i=1}^{N} (-1)^{b_i} \equiv \delta_{\mathbf a,\mathbf b} \lambda_{\mathbf a}.
\end{align}
Here we have introduced $\mathbf r\equiv \mathbf x \oplus \mathbf y$ and used $\sum_{\mathbf x} (-1)^{\mathbf x\cdot \mathbf c} = D\delta_{\mathbf c,\boldsymbol{0}}$. 
The eigenvalue associated to ``wave vector'' $\mathbf a\in \{0,1\}^N$ is 
\begin{equation}
    \lambda_{\mathbf a} = \frac{1}{N}\sum_{i=1}^{N} (1-2a_i) = 1-2\frac{|\mathbf a|}{N},
    \label{eq:app_srw_eigenvalues}
\end{equation}
which is only a function of Hamming weight $|\mathbf a|$. 

There is a non-degenerate $\lambda_0 = 1$ eigenvalue, associated to $\mathbf a = \boldsymbol{0}$ (the uniform state), then a gap of magnitude $2/N$, and then an $N$-fold degenerate eigenvalue $\lambda_1 = 1-2/N$. 
The relaxation time is therefore $t_{\rm rel} = 1/(1-\lambda_1) = N/2$.
One can also identify the mixing time exactly and prove that the problem (with idling) features a cutoff~\cite{diaconis_asymptotic_1990}. A simple way to estimate the mixing time is to 
require that $N\lambda_1^t$ be small, which gives $t_{\rm mix} \sim N\log(N)$ (the correct scaling up to constant factors).

The all-to-all gate set of Sec.~\ref{sec:non-local_gate} induces a simple random walk on the hypercube with idling parameter $p = 2^{1-N} = 2/D$, giving rescaled eigenvalues $\lambda_1 = 1-\frac{2}{N} \mapsto \frac{2}{D}\lambda_1 + 1 - \frac{2}{D} = 1-\frac{4}{DN}$, hence the relaxation time $t_{\rm rel} = ND/4$. 
Similarly the local gate set of Sec.~\ref{sec:local-circuit} induces a simple random walk on the hypercube with idling parameter $p = 1/4$, giving $\lambda_1 = 1-\frac{2}{N} \mapsto \frac{1}{4}\lambda_1 + \frac{3}{4} = 1-\frac{1}{2N}$, hence the relaxation time $t_{\rm rel} = 2N$. 

\subsection{Many particles: multipole eigenmodes} \label{app:hypercube_rw_multipole}

The eigenvalues $\lambda_{\mathbf a}$, Eq.~\eqref{eq:app_srw_eigenvalues}, also appear in the spectra of the Markov chains with $m>1$. 
We can construct the corresponding eigenmodes directly. Let us introduce the functions
\begin{align}
    f_{\mathbf a}(S) & = \frac{n_S(Z_{\mathbf a} = +1) - n_S(Z_{\mathbf a} = -1)}{m},
    \label{eq:generalized_eigenvector}
\end{align}
where for each bitstring $\mathbf a \in \{0,1\}^N$ we defined $Z_{\mathbf a} = \bigotimes_{i=1}^N Z_i^{a_i}$, and $n_S(Z_{\mathbf a} = \pm 1)$ is the number of elements $\mathbf z\in S$ such that $\bra{\mathbf z} Z_{\mathbf a} \ket{\mathbf z} = \pm 1$. 
The action of the transition matrix $\Gamma_{S,S'}$ on these functions is:
\begin{align}
    \sum_{S'} \Gamma_{S,S'} f_{\mathbf a}(S')
    & = \sum_{S'} \frac{| \{ u\in\mathcal G: u(S)=S'\}| }{|\mathcal G|} f_{\mathbf a}(S') \nonumber \\
    & = \sum_{u\in \mathcal G} \frac{1}{|\mathcal G|} f_{\mathbf a}(u (S) ) \label{eq:app_sum_uG}
\end{align}
Here we have used the definition of the transition probability $\Gamma_{S,S'}$ as the fraction of gates $u$ in the gate set $\mathcal{G}$ that maps subset $S$ to $S'$, and then changed variables via $S' \equiv u(S)$. Summing over $u$ over-counts each subset $S'$ in the original sum by a factor of $| \{ u'\in\mathcal G: u'(S)=u(S)\}|$, which precisely cancels the numerator in $\Gamma_{S,S'}$. 
Now we note that 
\begin{equation}
f_{\mathbf a}(u(S)) = \frac{2}{m}n_{u(S)}(Z_{\mathbf a}=1)-1 = \frac{2}{m} n_{S}(uZ_{\mathbf a}u=1)-1.
\end{equation}
Let $u\equiv u_{i,c_-,c_+}$ ($i$ is the location of the target, $c_\pm$ are the control bits for sites $i\pm 1$). If $a_i=0$, then $u Z_{\mathbf a}u= u^2 Z_{\mathbf a} = Z_{\mathbf a}$.
If $a_i=1$, we define $Z_{\mathbf a} \equiv Z_{\mathbf b} Z_i$, so that $uZ_{\mathbf a}u = Z_{\mathbf b} uZ_iu$. The operator $uZ_i u$ is still diagonal in the computational basis. A bitstring has $uZ_{\mathbf a}u=+1$ if and only if $Z_{\mathbf b} = uZ_iu$. 
Any given bitstring $\mathbf{z}\in S$ will change its value of $Z_i$ under gate $u_{i,c_-,c_+}$ for exactly one out of four possible choices for the control bits $c_\pm$ (namely $c_\pm = z_{i\pm 1}$). It follows that 
\begin{align}
    & \sum_{c_\pm = 0,1} n_S(Z_{\mathbf b} = u_{i,c_-,c_+}Z_iu_{i,c_-,c_+}=1 ) 
    = 3n_S(Z_{\mathbf b} = Z_i=1 ) + n_S(Z_{\mathbf b} = -Z_i=1 )
\end{align}
and
\begin{align}
    & \sum_{c_\pm = 0,1} n_S(Z_{\mathbf b} = u_{i,c_-,c_+}Z_iu_{i,c_-,c_+}=-1 ) 
    = 3n_S(Z_{\mathbf b} = Z_i=-1 ) + n_S(Z_{\mathbf b} = -Z_i=-1 ).
\end{align}
Adding the two gives
\begin{align}
    \sum_{c_\pm = 0,1} n_S(u_{i,c_-,c_+}Z_{\mathbf a} u_{i,c_-,c_+}=1 ) 
    & = 3n_S(Z_{\mathbf a} = 1 ) + n_S(Z_{\mathbf a} = -1 )
    \nonumber \\
    & = m+2n_S(Z_{\mathbf a} = 1 ),
\end{align}
where we used the fact that $n_S(Z_{\mathbf a}=1) + n_S(Z_{\mathbf a}=-1) = |S| = m$. 
Therefore, when we sum over the control bits $c_\pm \in \{0,1\}$ as part of our sum over $u\in \mathcal{G}$ in Eq.~\eqref{eq:app_sum_uG}, if $a_i = 1$ we find 
\begin{align}
    \sum_{c_\pm = \pm 1} f_{\mathbf a} (u_{i,c_-,c_+}(S))
    & = \frac{2}{m} [m+2n_S(Z_{\mathbf a}=1)]-4 
    = 2f_{\mathbf a} (S)
\end{align}
Combining this with the result for $a_i = 0$ gives
\begin{equation}
    \sum_{c_\pm = \pm 1} f_{\mathbf a}(u_{i,c_-,c_+}(S)) = 2(2-\delta_{a_i,1}) f_{\mathbf a}(S).
\end{equation}
Finally, plugging this into Eq.~\eqref{eq:app_sum_uG}, we obtain
\begin{align}
    \frac{1}{4N} \sum_{u\in \mathcal{G}} f_{\mathbf a}(u(S))
    & = \frac{1}{4N} \sum_{i=1}^N \sum_{c_\pm = \pm 1} f_{\mathbf a}(u_{i,c_-,c_+} (S)) \nonumber \\
    & = \frac{1}{N} \sum_{i=1}^N \frac{4-2\delta_{a_i,1}}{4} f_{\mathbf a}(S) \nonumber \\
    & = \left(1 - \frac{|\mathbf a|}{2N} \right) f_{\mathbf a}(S) \nonumber \\
    & = \lambda_{\mathbf a} f_{\mathbf a}(S),
\end{align}
where we recognize the eigenvalue $\lambda_{\mathbf a}$ of the simple random walk (with $3/4$ idling probability). 

These eigenmodes for the $m$-particle problem have an intuitive interpretation.
$|\mathbf a|=0$ represents a monopole moment (i.e. total particle number), which is conserved;
$|\mathbf a|=1$ yields the $N$ components of the dipole moment, or center of mass position, which are the slowest decaying modes in the problem;
$|\mathbf a|=2$ yields the $\binom{N}{2}$ components of the quadrupole moment, which decay more quickly; and so on for higher multipole moments. 

For $m=1$ the $2^N$ $\{f_{\mathbf a}\}$ modes exhaust the spectrum, whereas for $m\geq 2$ they make up only a small part. Still, they are very important, since they describe the decay of circuit-averaged expectation values $\overline{\langle Z_{\mathbf a} \rangle}$. 
Indeed, we have 
\begin{align}
    \overline{\langle Z_{\mathbf a} (t)\rangle} 
    & = \sum_S p_t(S) \bra{\psi_S} Z_{\mathbf a} \ket{\psi_S} 
    = (p_t|f_{\mathbf a}) 
    = (p_0|\Gamma^t|f_{\mathbf a}) 
    = \left(1-\frac{|\mathbf a|}{2N}\right)^t (p_0|f_{\mathbf a}),
    \label{eq:Z_ev_decay}
\end{align}
where we used the definition $\bra{\psi_S}Z_{\mathbf a} \ket{\psi_S} = f_{\mathbf a}(S)$, introduced the inner product $(a|b) = \sum_S a(S) b(S)$ between real functions on $\Sigma_m$, and used the Hermiticity of the transition matrix $\Gamma$ (reversibility of the Markov chain) to express the time-evolved probability distribution $(p_t| = (p_0|\Gamma^t$. 

It follows that any eigenmode $h(S)$ of the $m$-particle random walk whose eigenvalue is {\it not} in the single-particle spectrum does not show up in circuit-averaged expectation values $\overline{\langle Z_{\mathbf a}\rangle}$:
\begin{align}
    \sum_S h(S) \bra{\psi_S} Z_{\mathbf a} \ket{\psi_S}
    = \sum_S h(S) f_{\mathbf a}(S) = (h|f_{\mathbf a}) = 0,
\end{align}
due to orthogonality of the eigenmodes (Hermiticity of $\Gamma$). 
Such eigenmodes instead can appear in correlation functions across circuit realizations, e.g. $\overline{\langle Z_{\mathbf a} \rangle \langle Z_{\mathbf b} \rangle}$. 

\subsection{Two particles: relative coordinate eigenmodes} \label{app:hypercube_rw_relative}

Here we provide more details on eigenmodes of the two-particle problem that depend on subset $S = \{\mathbf z, \mathbf z'\}$ only through the relative coordinate $\mathbf r = \mathbf z \oplus \mathbf z'$. 
Let us consider how $\mathbf r$ might change under the action of our local automaton gate set $\mathcal G = \{u_{iab}\}$. There are two cases to consider:
\begin{itemize}
    \item If $r_{i-1} = r_{i+1} = 0$ (i.e., $z_{i-1} = z'_{i-1}$ and $z_{i+1} = z'_{i+1}$), then either both $z_i$ and $z_i'$ flip, or neither one flips; either way, $\mathbf r$ is unchanged.
    \item Otherwise, there is a $1/4$ probability (based on randomly drawn control bits $a,b$) that $u_{iab}$ flips $z_i$ but not $z_i'$, a $1/4$ probability that it flips $z_i'$ but not $z_i$, and a $1/2$ probability that it flips neither. In all, there is a $1/2$ probability that $r_i$ flips.
\end{itemize}
The rules above describe a random walk on $\mathbb{Z}_2^N$ based on ``activated'' bit flips, where $r_i$ can flip only if one of its neighbors $r_{i\pm 1}$ is 1. Clearly $\mathbf r = \boldsymbol{0}$ is a steady state, but it does not describe a valid element of $\Sigma_2$ (it gives $\mathbf z = \mathbf z'$). One can check that the restriction to $\mathbb{Z}_2^N\setminus \{\boldsymbol{0}\}$ is irreducible, aperiodic (since lazy moves are allowed) and reversible, thus it has the uniform distribution as unique steady state. 

One can write down the transition matrix $\Gamma|_R$ for the relative coordinate as a quantum Hamiltonian:
\begin{align}
    \Gamma|_R 
    & = \frac{1}{N} \sum_{i=1}^N \left[ \left(I - \Pi_{i} \right) \frac{I+X_i}{2} 
    + \Pi_{i}\right] \equiv I - \frac{H}{N}, \label{eq:pxp_like_H} \\
    \Pi_{i} & = \ketbra{0}_{i-1} \otimes \ketbra{0}_{i+1} \otimes I_{\text{not } i\pm 1}. 
\end{align}
The linear-in-$N$ relaxation time corresponds to the Hamiltonian $H$ being gapped. 
The Hamiltonian $H$ clearly has two ground states $\ket{0}^{\otimes N}$ (corresponding to the trivial steady state $\mathbf r = \boldsymbol{0}$) and $\ket{+}^{\otimes N}$ (corresponding to the uniform distribution); both states are annihilated by each individual term in $H$, so $H$ is frustration-free. One can use criteria for the existence of gaps in frustration-free, translationally invariant Hamiltonians~\cite{gosset_local_2016} to establish the presence of a gap as $N\to\infty$ from numerical evaluation of the gap in finite-sized systems. 
Such an approach reveals that $H$ is indeed gapped as $N\to\infty$~\cite{mittal_private_2024}. Note that while the criterion from Ref.~\cite{gosset_local_2016} is formulated for two-body Hamiltonians, it can be shown to apply for the three-body Hamiltonian $H$ in Eq.~\eqref{eq:pxp_like_H} because all next-nearest-neighbor terms commute~\cite{mittal_private_2024}. This establishes the linear-in-$N$ relaxation time for the relative-coordinate Markov chain $\Gamma|_R$, consistent with numerical observations in Fig.~\ref{fig:spectra}.


\section{Properties of the local automaton gate set}
\label{app:local_gateset}

In this Appendix we give more details on the local automaton gate set $\mathcal G$ used in Sec.~\ref{sec:local-circuit}. 
We recall that the gate set is defined as $\mathcal{G}=\{u_{iab} = \textsf{C}_{i-1,a}\textsf{C}_{i+1,b} {X}_i\}$, with $i\in [N]$, $a,b\in\{0,1\}$, and periodic boundary conditions (i.e., $i\pm 1$ are taken modulo $N$). $\textsf{C}_{j,a}$ denotes a control that is activated if bit $j$ is in state $a$. 
Each time step consists of drawing a gate $u_{iab}\in\mathcal{G}$ uniformly at random and applying it to the system. This defines a random walk on the permutation group $\perm{D}$ with transition matrix
\begin{equation}
    \Gamma_{\sigma,\sigma'} = \frac{1}{|\mathcal G|}\sum_{u\in\mathcal{G}} \delta_{\sigma', u\circ \sigma},
\end{equation}
where with slight abuse of notation we treat the gate $u$ as a permutation of the computational basis.
Through the natural action of permutations on subsets, 
$$
S \xrightarrow{\sigma} \{\sigma(i):\, i\in S\} \equiv \sigma(S),
$$ 
this also defines a Markov chain on the subset spaces $\Sigma_m$ for each $m$, with transition matrices
\begin{equation}
    \Gamma_{S,S'} = \frac{1}{|\mathcal G|} \sum_{u\in \mathcal G} \delta_{S', u(S)}
\end{equation}
(the dependence on $m$ is left implicit).

To show that $\Gamma_{S,S'}$ equilibrates to the unique steady state $\pi_m$ (the uniform distribution on $\Sigma_m$), we need the following three facts to hold:

\begin{lemma}[reversibility]
    The Markov chain $\Gamma_{S,S'}$ on $\Sigma_m$ is reversible for all $m$.
\end{lemma}

\begin{proof}
    Since the gate set $\mathcal{G}$ is such that $u^2 = I$ for all $u\in\mathcal{G}$, we have $u(S)=S'$ if and only if $u(S') = S$. Thus $\Gamma_{S,S'} = \Gamma_{S',S}$ for all $S,S'\in\Sigma_m$. 
\end{proof}

\begin{lemma}[irreducibility]
    The Markov chain $\Gamma_{S,S'}$ on $\Sigma_m$ is irreducible for all $m \leq D - 2$. 
\end{lemma}

\begin{proof}
    It is well known~\cite{shende_synthesis_2003,aaronson_classification_2015} that the Toffoli ($\CCX$) and NOT ($X$) gates are ``almost'' universal for reversible classical computation, in the sense that they generate the subgroup of $\perm{D}$ made of even permutations (permutations that can be decomposed into an even number of transpositions). 
    Our gate set includes Toffoli gates between three consecutive qubits, with the controls on the sides and the target in the middle ($u_{i11}$). Further, it can generate NOT gates ($X_i = u_{i00} u_{i01} u_{i10} u_{i11}$) and CNOT gates between neighboring qubits, with either choice of control and target: 
    $$ \textsf{C}_{i-1,a}X_i = u_{ia0} u_{ia1}, \qquad \textsf{C}_{i+1,a} X_i = u_{i0a} u_{i1a}. $$
    Witht the local CNOTs we can make a local SWAP gate: 
    $$\textsf{SWAP}_{i,i+1} = \textsf{C}_{i,0}{X}_{i+1} \circ \textsf{C}_{i+1,0} {X}_i \circ \textsf{C}_{i,0} {X}_{i+1}. $$
    Finally, by composing local SWAPs we can make arbitrary permutations of the qubits, and thus generate the whole (non-geometrically-local) Toffoli + NOT gate set. Therefore $\mathcal{G}$ generates all even permutations.
    Now, given two sets $S,S'\in\Sigma_m$, we can find a permutation $\sigma\in\perm{D}$ that maps one to the other: $S' = \sigma(S)$. If $\sigma$ is even, we are done. If $\sigma$ is odd, then let $i,j\in [D]$ be two basis states that are {\it not} in subset $S$ (we have assumed $|S| = m \leq D-2$ so that two such states exist), and let $\tau_{ij}\in\perm{D}$ be the transposition $i\leftrightarrow j$. Then $u\circ \tau_{ij}$ is an even permutation that maps $S$ to $S'$.
    Since any two sets can be connected by an even permutation and any even permutation can be decomposed into a product of $u\in \mathcal{G}$, the Markov chain can connect any two subsets $S,S'\in\Sigma_m$ and is thus irreducible.
\end{proof}

\begin{lemma}[aperiodicity]
    The Markov chain $\Gamma_{S,S'}$ on $\Sigma_m$ is aperiodic for any $m \leq D / 4$. 
\end{lemma}

\begin{proof}
    Given two sets $S,S'\in\Sigma_m$, we are going to show that the Markov chain can connect $S$ to $S'$ with paths of any length $l\geq l_{\rm min}$, with $l_{\rm min}$ some finite integer.
    To do so, let us introduce a third state $S''$ with the property that all bitstrings $z\in S''$ have zeros at positions $i-1$ and $i+1$, for some $i$ (this limits the size of the subset to $m\leq D/4$). 
    Since the chain is irreducible, we can find a sequence of $l_1$ gates whose product maps $S$ to $S''$ and another sequence of $l_2$ gates whose product maps $S''$ to $S'$. Then, inserting a gate $u_{i11}$ between these two sequences $\iota$ times produces a path from $S$ to $S'$ of length $l_1+\iota+l_2$. This can take any value larger than $l_{\rm min} \equiv l_1+l_2$. 
\end{proof}

Reversibility, irreducibility and aperiodicity imply that the Markov chain has the uniform distribution $\pi_m$ as its unique steady state \cite{levin_markov_2017}. This proves that our 3-qubit automaton gate set equilibrates to the uniform distribution on all subset states with $|S| = m \leq D/4$. 

It is easy to see that three-qubit gates are necessary to obtain the above result. Indeed two-qubit automaton gates (generated by NOT and CNOT) are {\it affine transformations} of the computational basis, in the sense that for any such gate $u$ we have $u\ket{\mathbf z} = \ket{A_u \mathbf z \oplus \mathbf{b}_u}$, with $A_u$ a binary matrix and $\mathbf{b}_u$ a bitstring. Any transformations that are not affine cannot be generated by two-qubit automaton gates. As it turns out, the gate set comprising NOT and CNOT gates induces an irreducible Markov chain only up to $m = 3$; an example of two subsets with $m = 4$ and $N = 3$ that cannot be connected by NOTs and CNOTs only is given by 
\begin{equation}
    S = \begin{bmatrix}
    0 & 1 & 1 \\ 
    0 & 1 & 0 \\ 
    0 & 0 & 1 \\
    0 & 0 & 0
    \end{bmatrix}, 
    \quad
    S' = \begin{bmatrix}
    1 & 0 & 0 \\ 
    0 & 1 & 0 \\ 
    0 & 0 & 1 \\
    0 & 0 & 0
    \end{bmatrix}
\end{equation}
(each row is a bitstring $\mathbf z$ in the set). It is easy to see that the sum (modulo 2) of bitstrings in each set, $\mathbf z_{\rm tot}$, transforms under an affine gate $u$ according to $\mathbf{z}_{\rm tot} \mapsto A_u \mathbf{z}_{\rm tot}$ (the additive term $\mathbf{b}_u$ appears 4 times and thus cancels). But $S$ has $\mathbf{z}_{\rm tot} = (0,0,0)$ while $S'$ has $\mathbf{z}_{\rm tot} = (1,1,1)$, so any sequence of $u$'s mapping $S$ to $S'$ must be such that $\prod_u A_u$ annihilates $(1,1,1)$. This is incompatible with a reversible computation (e.g. it would imply a collision between strings $000$ and $111$). 

Another way of arriving at a similar conclusion is to note that pseudoentanglement requires magic~\cite{grewal_stabilizer_2023,grewal_pseudoentanglement_2024}, and two-qubit automaton gate sets happen to also be Clifford. Thus if we start from a stabilizer subset state (e.g. $\ket{0}^{\otimes N-k}\otimes \ket{+}^{\otimes k}$) we cannot produce a pseudoentangled ensemble. 
However, for $m\leq 3$, it is instead possible to show that the gate set comprising NOT and CNOT yields irreducible Markov chains. This is perhaps related to the fact that Clifford circuits form 3-designs on qubits, so matching the 3 lowest moments of the Haar measure does not require magic.


\section{Upper bound on relaxation time with local gate set}
\label{app:local_relaxation}

In Sec.~\ref{sec:local-circuit} we study the relaxation time of the Markov chains induced by the local automaton gate set $\mathcal{G} = \{ \textsf{C}_{i-1,a}\textsf{C}_{i+1,b}{X}_i \}$ and, based on numerical observations, we conjecture the scaling $t_{\rm rel}^{(m)} = \Theta(N)$, independent of $m$ (Conjecture~\ref{conjecture:relaxation}). 
While the single-particle spectrum, solved analytically in App.~\ref{app:hypercube_rw}, gives a lower bound $t_{\rm rel}^{(m)} \geq \Omega(N)$, an upper bound is harder to obtain analytically. 
Here we prove the upper bound
\begin{equation}
    t_{\rm rel}^{(m)} \leq O(mN^5)
    \label{eq:app_trel_upperbound}
\end{equation}
based on the {\it comparison method}~\cite{levin_markov_2017} applied to ``simple permutations''~\cite{gowers_almost_1996,brodsky_simple_2008}, a set of 3-local but not geometrically local reversible classical gates\footnote{
Based on the results of Ref.~\cite{he_pseudorandom_2024}, which appeared after the first version of this work, the comparison method applied to a set of geometrically-local 3-bit classical gates known as $\textsf{DES[2]}$ (which is a super-set of our gate set $\mathcal{G}$) would yield an improved bound $t_{\rm rel}\leq O(mN)$.
}. This is very loose compared to Conjecture~\ref{conjecture:relaxation} ($t_{\rm rel}\leq O(N)$) but serves to rigorously bound above [via Eq.~\eqref{eq:mix_eigenbound}] the mixing time of our Markov chains, and thus the $m$-copy indistinguishability times $\pseudothermtime{m}$, as $O(m\, \poly{N})$. 

The comparison method bounds the relaxation time $t_{\rm rel}$ of a Markov chain $\Gamma$ in terms of the relaxation time $\tilde{t}_{\rm rel}$ of another Markov chain $\tilde{\Gamma}$ on the same space as~\cite{levin_markov_2017}
\begin{equation}
t_{\rm rel} \leq B \tilde{t}_{\rm rel}, \label{eq:comparison_bound}
\end{equation}
in terms of a ``congestion ratio'' $B$.
To define $B$, we first need to introduce some background and notation. 
We consider a graph $(V,E)$ with vertices $V$ and edges $E$; the reversible Markov chain $\Gamma$ allows transitions between any two vertices connected by an edge, with probability $\Gamma_{e_1 e_2}$ where $(e_1,e_2) = e\in E$ denotes the edge. We also have another graph with the same vertices but different edges, $(V,\tilde{E})$, with a Markov chain $\tilde{\Gamma}$ allowing transitions along the $\tilde{E}$ edges. 
We assume that, for any edge $\tilde{e} = (x,y) \in\tilde{E}$ in the second graph, one can find a path $\gamma_{xy}$ along the first graph $(V,E)$ connecting $x$ to $y$. $|\gamma_{xy}|$ denotes the length of that path. We say that $e\in \gamma_{xy}$ if the path traverses edge $e \in E$. With this notation in place, we can write down the expression for the congestion ratio $B$:
\begin{equation}
    B = \max_{e\in E}\frac{1}{\Gamma_{e_1e_2}}\sum_{x,y:\gamma_{xy}\ni e} \tilde{\Gamma}_{xy} \left|\gamma_{xy}\right|. \label{eq:comparison_coeff_general}
\end{equation}

Our strategy is to use our Markov chain (induced by the local automaton gate set $\mathcal{G}$) as $\Gamma$ and use the Markov chain induced by {\it `simple permutations'}~\cite{gowers_almost_1996,hoory_simple_2005} as $\tilde{\Gamma}$. `Simple permutations' are defined as the gate set $\tilde{\mathcal G} = \{v_{ijk,\sigma}:\, i,j,k\in [N],\, \sigma\in S_8\}$, where $v_{ijk,\sigma}$ acts trivially on all bits not in $\{i,j,k\}$, and as a permutation $\sigma\in S_8$ on bits $\{i,j,k\}$. We further restrict $\sigma\in S_8$ to be a `width-2' simple permutation, i.e. a permutation that leaves two of the three input bits unchanged~\cite{hoory_simple_2005,brodsky_simple_2008} (the Toffoli gate falls in this category, so $\mathcal{G}\subset\tilde{\mathcal{G}}$). 
Note that the bits $i,j,k$ need not be consecutive, so `simple permutations' are 3-local but not geometrically local. 
`Simple permutations' have been studied in the probability literature~\cite{gowers_almost_1996,hoory_simple_2005} due to their good mixing properties, and several bounds on their relaxation and mixing times are known. 
For our purpose, we will need the following fact:

\begin{fact}[relaxation time from simple permutations] \label{fact:app_simpleperm}
    `Width-2 simple permutations' $\tilde{\mathcal G}$ induce a Markov chain on $\Sigma_m$ whose relaxation time is bounded above as
    \begin{equation}
        \tilde{t}_{\rm rel}^{(m)} \leq O(mN^2). \label{eq:app_simpleperm_relaxation}
    \end{equation}
\end{fact}

\begin{proof}
    Ref.~\cite{brodsky_simple_2008} considers the random walk induced by `simple permutations' on {\it ordered $m$-tuples}, $\mathbf z = (z_1,\dots z_m)$, rather than sets $S = \{z_1,\dots z_m\}\in \Sigma_m$. It shows that the former has relaxation time $\tilde{t}_{\rm rel}^{(m,\text{ord.})} \leq O(mN^2)$. 
    We will show that 
    \begin{equation} 
    \tilde{t}_{\rm rel}^{(m)} \leq \tilde{t}_{\rm rel}^{(m,\text{ord.})}
    \label{eq:app_ordered_vs_unordered}
    \end{equation} 
    (the walk on ordered $m$-tuples relaxes no faster than that on subsets), which implies Eq.~\eqref{eq:app_simpleperm_relaxation}.
    
    Consider an eigenfunction $f$ of the Markov chain on subsets:
    \begin{equation}    
    \frac{1}{|\tilde{\mathcal G}|} \sum_{v\in\tilde{\mathcal G}} f(v(S)) = \lambda f(S).
    \end{equation}
    Let us define a function $g$ on the space of ordered $m$-tuples by setting $g(\mathbf z) = f({\rm set}(\mathbf z))$, where ${\rm set}(\mathbf z) \in \Sigma_m$ is the {\it unordered} set of elements in the $m$-tuple $\mathbf z$. We can easily check that $g$ is an eigenfunction of the Markov chain on ordered $m$-tuples with the same eigenvalue $\lambda$:
    \begin{align}
        \frac{1}{|\tilde{\mathcal G}|} \sum_{v\in\tilde{\mathcal G}} g(v(\mathbf z))
         & = \frac{1}{|\tilde{\mathcal G}|} \sum_{v\in\tilde{\mathcal G}} f({\rm set}(v(\mathbf z))) \nonumber \\
         & = \frac{1}{|\tilde{\mathcal G}|} \sum_{v\in\tilde{\mathcal G}} f(v({\rm set}(\mathbf z))) \nonumber \\
         & = \lambda f({\rm set}(\mathbf z)) = \lambda g(\mathbf z). 
    \end{align}
    Thus the spectrum of the Markov chain on subsets is included within the spectrum of the Markov chain on ordered tuples. Recalling that the relaxation time is the inverse spectral gap, Eq.~\eqref{eq:app_ordered_vs_unordered} follows.
\end{proof}

Now that we have an upper bound on $\tilde{t}_{\rm rel}^{(m)}$, it remains to upper-bound the coefficient $B$ in the comparison bound, Eq.~\eqref{eq:comparison_coeff_general}.
We break it down into three bounds:

\begin{fact} \label{fact:app_Bbound1}
For any edge $e\in E$, we have
    \begin{equation} 
    \frac{1}{\Gamma_e} \leq O(N). \label{eq:app_Bbound1} 
    \end{equation}
\end{fact}

\begin{proof}
Follows from the fact that $\Gamma_{S,S'} \geq 1/|\mathcal G| = \frac{1}{4N}$ (at least one gate $u\in \mathcal{G}$ must connect $S$ to $S'$). 
\end{proof}

\begin{fact} \label{fact:app_Bbound2}
For any two subsets $S$, $S'$ we have
\begin{equation}
    |\gamma_{S,S'}| \leq O(N). \label{eq:app_Bbound2}
\end{equation}
\end{fact}

\begin{proof}
For any two sets $S$, $S'$ connected by a simple permutation $v_{ijk,\sigma}\in\tilde{\mathcal G}$, we can choose a path of local automaton gates $u\in\mathcal{G}$ that implements a sequence of SWAPs to bring qubits $i,j,k$ to consecutive positions (this takes $O(N)$ gates), implement the permutation $\sigma$ (this takes $O(1)$ gates), and then return the three qubits to positions $i,j,k$ ($O(N)$ gates again). In all this path has length $|\gamma_{S,S'}| \leq O(N)$. 
\end{proof}

\begin{fact} \label{fact:app_Bbound3}
For any edge $e\in E$ we have
\begin{equation}
    \sum_{S,S':\, \gamma_{S,S'} \ni e} \tilde{\Gamma}_{S,S'} \leq O(N). \label{eq:app_Bbound3}
\end{equation}
\end{fact}

\begin{proof}
We can write $S'$ as $v(S)$ for some gate $v\in\tilde{\mathcal G}$. There may be more than one way to do this; to account for it, we introduce the multiplicity $\mu_{S,S'} \equiv |\{ v\in \tilde{\mathcal G}:\, S' = v(S)\}|$. We obtain 
\begin{equation}
    \sum_{\substack{S,S'\in\Sigma_m: \\ \gamma_{S,S'} \ni e}} \tilde{\Gamma}_{S,S'}
    = \sum_{\substack{v\in\tilde{\mathcal G}, S\in \Sigma_m : \\ \gamma_{S,v(S)}\ni e}} \frac{\tilde{\Gamma}_{S,v(S)}}{\mu_{S,v(S)}} 
\end{equation}
where the factor $1/\mu_{S,v(S)}$ accounts for the over-counting of terms when switching summation variables from $S'$ to $v$. 
Now, by definition $\tilde{\Gamma}_{S,S'} = \mu_{S,S'} / |\tilde{\mathcal G}|$ (fraction of the gate set that maps $S$ to $S'$), which gives
\begin{equation}
    \sum_{\substack{S,S'\in\Sigma_m: \\ \gamma_{S,S'} \ni e}} \tilde{\Gamma}_{S,S'}
    = \frac{1}{|\tilde{\mathcal{G}}|} | \{ (v,S):\, \gamma_{S,v(S)} \ni e \} |. \label{eq:app_Bbound3_midstep}
\end{equation}
Our choice of paths $\gamma_{S,S'}$ (see proof of Eq.~\eqref{eq:app_Bbound2} above) depends only on the gate $v$: 
we decompose $v = u_l u_{l-1} \cdots u_1$ into a sequence of $l\leq O(N)$ local automaton gates, $u_i\in \mathcal{G}$. The condition $\gamma_{S,v(S)} \ni e$ thus requires that $e_1 = u_i \cdots u_1(S)$ and $e_2 = u_{i+1}(e_1)$ for some $i \in [l]$; equivalently, it requires $S = u_1 \cdots u_i (e_1)$. Therefore, given $e$ and $v$, there are $O(N)$ choices for $S$, indexed by the integer $i$. The $|\tilde{\mathcal G}|$ choices for $v$ cancel out the denominator in Eq.~\eqref{eq:app_Bbound3_midstep}, and we obtain the desired bound. 
\end{proof}

Putting it all together, Fact~\ref{fact:app_simpleperm} bounds the relaxation time in the `simple permutation' Markov chain as $\tilde{t}_{\rm rel}^{(m)} \leq O(mN^2)$;
Facts~\ref{fact:app_Bbound1}, \ref{fact:app_Bbound2}, and \ref{fact:app_Bbound3} together yield the congestion ratio $B\leq O(N^3)$; 
plugging these into the comparison bound [Eq.~\eqref{eq:comparison_bound}] gives the desired result [Eq.~\eqref{eq:app_simpleperm_relaxation}].


\section{Upper bound on $m$-copy indistinguishability time for subset states}
\label{app:subset_ptherm_bound}

In this Appendix we prove the upper bound Eq.~\eqref{eq:tptherm_upperbound} on the $m$-copy indistinguishability time for subset states, Sec.~\ref{sec:generation}. We start from the key result, Eq.~(\ref{eq:subsetstate_mapping1}-\ref{eq:subsetstate_mapping2}), which relates ensemble pseudothermalization to the equilibration of the Markov chains $\Phi_{m+\delta \leftarrow K}[p_t]$ on $\Sigma_{m+\delta}$, for $\delta = 0,\dots m$. 

Let us decompose the matrix $\mathcal{M}^{(m)}_{p_t}$ based on the values of the difference between subsets, $\delta = |S'\setminus S''|$:
\begin{align}
    \mathcal{M}^{(m)}_{p_t} 
    & = \sum_{\delta = 0}^m O_\delta, \\
    O_\delta
    & = \binom{K}{m}^{-1} \binom{K}{m+\delta}
    \sum_{\substack{S',S''\in \Sigma_m:\\ |S'\setminus S''| = \delta}} \ket{S'} f^{(m+\delta)}_{p_t} (S'\cup S'') \bra{S''},
\end{align}
where $f^{(m+\delta)}_{p_t}$ is defined in Eq.~\eqref{eq:f_mdelta_def}. The diagonal term $O_0$ provides a lower bound: $\| \mathcal{M}^{(m)}_{p_t}\|_{\rm tr} \geq \| O_0 \|_{\rm tr}$, which is used to derive the lower bound on the $m$-copy indistinguishability time, Eq.~\eqref{eq:tptherm_lowerbound}. 
To get an upper bound, we need to consider the off-diagonal terms as well. 

We have
\begin{equation}
    \left\| \mathcal{M}_{p_t}^{(m)} \right\|_{{\rm tr}}
    \leq \sum_{\delta=0}^m \|O_{\delta}\|_{{\rm tr}}.
\end{equation}
We will upper-bound each $\|O_\delta\|_{\rm tr}$ by the sum of the absolute values of its elements: for any operator $A$, one has
\begin{equation}
    \|A\|_{\rm tr}
    = \left\| \sum_{ij} A_{ij} \ketbra{i}{j} \right\|_{\rm tr} 
    \leq \sum_{ij} \| A_{ij} \ketbra{i}{j} \|_{\rm tr} 
    = \sum_{ij} |A_{ij}|.
\end{equation}
For each $\delta$, the set of possible entries is give by $f_{m+\delta}(S)$, with $S\in\Sigma_{m+\delta}$. Each entry appears multiple times: there are $\binom{m+\delta}{m} \binom{m}{m-\delta}$ ways of choosing $S',S''\in \Sigma_m$ such that $|S'\cup S''| = m+\delta$. Therefore,
\begin{align}
    \| O_\delta \|_{\rm tr}
    & \leq \sum_{S\in \Sigma_{m+\delta}} \binom{m+\delta}{m} \binom{m}{m-\delta} \binom{K}{m}^{-1} \binom{K}{m+\delta}  |f^{(m+\delta)}_{p_t} (S)| \nonumber \\
    & \leq  \binom{m}{\delta} \binom{K-m}{\delta}
    \sum_{S\in \Sigma_{m+\delta}} |\Phi_{m+\delta\leftarrow K}[p_t](S)-\pi_{m+\delta}(S) | \nonumber \\
    & = 2\binom{m}{\delta} \binom{K-m}{\delta} \Delta_{m+\delta},
\end{align}
where in the second line we combined the various binomial coefficients and in the last line we recognized the total variation distance:
\begin{equation}
    \Delta_{m+\delta} = \frac{1}{2} \sum_{S\in\Sigma_{m+\delta}}|\Phi_{m+\delta\leftarrow K}[p_t](S) - \pi_{m+\delta}(S)|.
\end{equation}

Now, using the results of Appendix~\ref{app:phi_maps}, and specifically Fact~\ref{fact:monotonicity}, we have that $\Delta_{m+\delta}$ is non-decreasing in $\delta$. It follows that 
\begin{align}
    \frac{1}{2} \| \mathcal{M}^{(m)}_{p_t} \|_{\rm tr} 
    & \leq \Delta_{2m} \sum_{\delta = 0}^m \binom{m}{\delta} \binom{K-m}{\delta}
    = \binom{K}{m} \Delta_{2m}. 
\end{align}
Thus, to make the left hand side smaller than $\epsilon$, the Markov chain on $\Sigma_{2m}$ must equilibrate to within total variation error $\Delta_{2m} \leq \binom{K}{m}^{-1} \epsilon$.
Defining the mixing time $t_{\rm mix}(\epsilon)$ based on a total variation threshold $\epsilon$ (see Sec.~\ref{sec:review_cutoff}), one has the following general inequality:
\begin{equation}
    t_{{\rm mix}}(\epsilon = 2^{-l-1})\le lt_{{\rm mix}}(\epsilon = 1/4).
\end{equation}
It follows then that (taking the mixing time based on $\epsilon = 1/4$ as usual)
\begin{align}
    \pseudothermtime{m} \leq \left[ \log_2 \binom{K}{m} \right] t_{\rm mix}^{(2m)}
    \leq m \log_2(K) t_{\rm mix}^{(2m)}.
\end{align}
Finally, invoking the upper bound on the mixing time (conditional on Conjecture~\ref{conjecture:relaxation}) gives
\begin{align}
    \pseudothermtime{m} \leq O(\log(K) m^2 N^2).
\end{align}
In the case of $\log K = {\rm polylog}(N)$ (where the entropy of the random subset ensemble is polylogarithmic) this reduces to 
\begin{align}
    \pseudothermtime{m} \leq \tilde{O}(m^2 N^2),
\end{align}
where the symbol $\tilde{O}$ represents an upper bound up to ${\rm polylog}(N)$ factors.

\end{document}